%% file: main.tex
\documentclass[12pt]{article}
\usepackage[utf8]{inputenc}
\usepackage[T1]{fontenc}
\usepackage[english]{babel}

\usepackage{amsmath}
\usepackage{amssymb}
\usepackage{amsfonts}
\usepackage{amsthm}
\usepackage{stmaryrd}
\usepackage[retainorgcmds]{IEEEtrantools}

\usepackage[english, onelanguage, linesnumbered, lined, algoruled]{algorithm2e}
\SetKwComment{Comment}{//}{}

 \newtheorem{theo}{Theorem}[section]
  \newtheorem{theorem}[theo]{Theorem}
 \newtheorem{prop}[theo]{Proposition}
 \newtheorem{proposition}[theo]{Proposition}
 \newtheorem{cor}[theo]{Corollary}
 
 \newtheorem{definition}[theo]{Definition}
 
 \newtheorem{lemma}[theo]{Lemma}
\theoremstyle{plain} 
\theoremstyle{plain} \newtheorem{remark}[theo]{Remark}
\theoremstyle{plain} \newtheorem{problem}[theo]{Problem}

\newcommand{\Z}{\mathbb{Z}}
\newcommand{\ZqZ}{\Z/q\Z}

\newcommand{\eqdef}{\stackrel{\text{def}}{=}}

\newcommand{\Ceil}[1]{\left\lceil #1 \right\rceil}
\newcommand{\Ceiling}[1]{\Ceil{#1}}
\newcommand{\Floor}[1]{\left\lfloor #1 \right\rfloor}

\newcommand{\F}{\mathbb{F}_2}
\newcommand{\Fm}{\mathbb{F}_{2^m}}
\newcommand{\Fq}{\mathbb{F}_q}

\newcommand{\fq}{\Fq}
\newcommand{\Fqm}{\mathbb{F}_{q^m}}

\newcommand{\fqm}{\Fqm}
\newcommand{\CG}[2]{\begin{bmatrix}#1 \\ #2\end{bmatrix}_q}
\newcommand{\cg}[2]{\CG{#2}{#1}}
\newcommand{\CGdeux}[2]{\begin{bmatrix}#1 \\ #2\end{bmatrix}_2}
\newcommand{\cgdeux}[2]{\CGdeux{#2}{#1}}

\newcommand{\OO}[1]{\mathcal{O}\left( #1 \right)}

\newcommand{\C}{{\mathcal{C}}}

\newcommand{\rand}{\stackrel{\$}{\leftarrow}}

\newcommand{\word}[1]{\ensuremath{\boldsymbol{#1}}}
\newcommand{\Av}{\word{A}}
\newcommand{\Bv}{\word{B}}

\newcommand{\Dv}{\word{D}}

\newcommand{\Gv}{\word{G}}
\newcommand{\Hv}{\word{H}}
\newcommand{\Iv}{\word{I}}

\newcommand{\Mv}{\word{M}}

\newcommand{\av}{\word{a}}
\newcommand{\bv}{\word{b}}
\newcommand{\cv}{\word{c}}

\newcommand{\ev}{\word{e}}

\newcommand{\gv}{\word{g}}
\newcommand{\hv}{\word{h}}

\newcommand{\mv}{\word{m}}

\newcommand{\sv}{\word{s}}

\newcommand{\uv}{\word{u}}

\newcommand{\xv}{\word{x}}
\newcommand{\yv}{\word{y}}

\newcommand{\sigmav}{\word{\sigma}}

\DeclareMathOperator{\Rank}{Rank}

\DeclareMathOperator{\Supp}{Supp}
\DeclareMathOperator{\Prob}{Prob}

\newcommand{\norme}[1]{\| #1 \|}
\newcommand{\prob}[1]{{\rm Prob}\left( #1 \right)}

\newcommand{\RSD}{\mathrm{RSD}}

\newcommand{\IRSD}{\mathrm{I-RSD}}
\newcommand{\IRSR}{\mathrm{I-RSR}}
\newcommand{\DUPKS}{\mathrm{I-LRPC}}

\newcommand{\keygen}{\mathrm{KeyGen}}
\newcommand{\encrypt}{\mathrm{Enc}}
\newcommand{\decrypt}{\mathrm{Dec}}

\newcommand{\Adv}{\ensuremath{\mathsf{Adv}}}
\newcommand{\mA}{\ensuremath{\mathcal{A}}}

\usepackage{booktabs}
\usepackage{dsfont} %for \mathds{1}_{n_2}
\usepackage{multirow}
\usepackage{graphicx}
\usepackage{placeins}
\usepackage[top=90pt, bottom=90pt, left=75pt , right=75pt, headsep=12pt]{geometry}
\usepackage{pdf14}
\usepackage[pagebackref]{hyperref}
\newcommand{\mapolicebackref}[1]{\mbox{\textsl{\small #1}}}
\renewcommand*{\backref}[1]{}
\renewcommand*{\backrefalt}[4]{%
\ifcase #1 \mapolicebackref{Uncited in this paper}
    \or \mapolicebackref{#2}
    \else \mapolicebackref{#2}
\fi
}

\hypersetup{
    unicode=false,          % non-Latin characters in Acrobat’s bookmarks
    pdftoolbar=true,        % show Acrobat’s toolbar?
    pdfmenubar=true,        % show Acrobat’s menu?
    pdffitwindow=false,     % window fit to page when opened
    pdfstartview={FitH},    % fits the width of the page to the window
    pdftitle={Algorithm Specifications and Supporting Documentation},    % title
    pdfauthor={
    },     % author
    pdfsubject={
    },   % subject of the document
    pdfcreator={
    },   % creator of the document
    pdfproducer={
	}, % producer of the document
    pdfkeywords={NIST} {Post-Quantum Cryptography} {Quantum-Safe Cryptography} {Code-based Cryptography} {Key Exchange}, % list of keywords
    pdfnewwindow=true,      % links in new window
    colorlinks=true,       % false: boxed links; true: colored links
}

\title{\bf{LAKE}\\ -a Low rank parity check codes Key Exchange
proposal for the NIST call-}
\author{}

\usepackage{marvosym}
\usepackage{wasysym}
\usepackage{pdfpages}

\newcommand{\KEM}{\mathsf{KEM}}
\newcommand{\KeyGen}{\normalfont\textsf{KeyGen}}
\newcommand{\INDCPA}{\mathsf{IND\mbox{-}CPA}}
\newcommand{\Encap}{\normalfont\textsf{Encap}}
\newcommand{\Decap}{\normalfont\textsf{Decap}}
\newcommand{\INDCPArand}{\INDCPA_{\mathsf{rand}}}
\newcommand{\INDCPAreal}{\INDCPA_{\mathsf{real}}}
\newcommand{\advA}{\mathcal{A}} % Adversary 
\newcommand{\AdvINDCPA}[2]{\Adv^{\normalfont{indcpa}}_{#1}(#2)}
\newcommand{\Exp}{\mathbf{Exp}}
\newcommand{\param}{\ensuremath{\mathsf{param}}}
\newcommand{\E}{\ensuremath{\mathcal{E}}\xspace}
\newcommand{\sk}{\ensuremath{\mathsf{sk}}\xspace}
\newcommand{\pk}{\ensuremath{\mathsf{pk}}\xspace}
\newcommand{\A}{\ensuremath{\mathcal{A}}\xspace}
\newcommand{\ind}{\normalfont\textsf{ind}}
\newcommand{\seck}{{\lambda}}
\newcommand{\sets}{\leftarrow}
\newcommand{\Setup}{\ensuremath{\mathsf{Setup}}}
\newcommand{\comreturn}{\texttt{RETURN\ }}

\newcommand{\GUESS}{\texttt{GUESS}}

\newcommand{\vect}[1]{\langle #1 \rangle}
\newcommand{\EF}{EF}
\title{Low Rank Parity Check Codes: New Decoding Algorithms and Applications to Cryptography}
\author{Nicolas Aragon, Philippe Gaborit, Adrien Hauteville,\\ Olivier Ruatta and Gilles Z\'emor}
\date{}
\begin{document}
\maketitle
\begin{abstract}
We introduce a new family of rank metric codes: Low Rank Parity Check codes (LRPC),
for which we propose an efficient probabilistic decoding algorithm. This family of codes can be seen as the
equivalent of classical LDPC codes for the rank metric. We then use these codes to design cryptosystems 
à la McEliece: more precisely we propose two schemes for key encapsulation mechanism (KEM) and public key
encryption (PKE). Unlike rank metric codes used in
 previous encryption algorithms 
-notably Gabidulin codes - LRPC codes have a very weak algebraic structure. Our cryptosystems can be
seen as an equivalent of the NTRU cryptosystem (and also to the more recent MDPC \cite{MTSB12} cryptosystem) in
a rank metric context. The present paper is an extended version of the article introducing LRPC codes, with important new 
contributions. We have improved the decoder thanks to a new approach
which allows for decoding of errors of higher rank weight, 
namely
up to $\frac{2}{3}(n-k)$ when
the previous decoding algorithm only decodes up to $\frac{n-k}{2}$ errors. Our
codes therefore outperform the classical Gabidulin 
code decoder which deals with weights up to $\frac{n-k}{2}$. This
comes at the expense of probabilistic decoding, but the decoding error probability
can be made arbitrarily small. 
The new approach can also be used to decrease the decoding error
probability of previous schemes, which is especially useful
for cryptography. Finally, we introduce ideal rank codes, which generalize double-circulant
rank codes and allow us to avoid known structural attacks based on folding. 
To conclude, we propose different parameter sizes for our schemes and we obtain a public key of 3337 bits for key exchange and 5893 bits for public key encryption, both for 128 bits of security.
\end{abstract}

\section{Introduction}
In recent years there has been a burst of activity in the field of post-quantum cryptography,
whose appeal has become even more apparent since the recent attacks on
the discrete logarithm problem in small characteristic \cite{BGJT14}. These events
stress that unexpected new attacks on classical
cryptographic systems can appear at any time and that it is important not to have all one's eggs in the same basket.

Among potential candidates for alternative cryptography, lattice-based and code-based cryptography
are strong candidates. Rank-based cryptography relies on the difficulty of decoding error-correcting
codes embedded in a rank metric space (often over extension fields of
fields of prime order), 
when code-based cryptography relies on difficult
problems related to error-correcting codes embedded in Hamming metric spaces (often over small fields
$\Fq$) and when lattice-based cryptography is mainly based on the study of $q$-ary lattices,
which can be seen as codes over rings of type $\ZqZ$ (for large $q$),
embedded in Euclidean metric spaces.

The particular appeal of the rank metric is that the practical
difficulty of the decoding problems
grows very quickly with the size of parameters. In particular, it is possible to reach a complexity
of $2^{80}$ for random instances with size only a few thousand bits, while for lattices or codes,
at least a hundred thousand bits are needed. Of course with codes and lattices it is possible
to decrease the size to a few thousand bits but with additional structure like
quasi-cyclicity \cite{BCGO09}, which comes at the cost of losing
reductions to known difficult problems.
The rank metric was introduced by Delsarte and Gabidulin \cite{G85}, along
with Gabidulin codes which are a rank-metric equivalent of Reed-Solomon
codes.
Since then, rank metric codes have been used for many applications:
coding theory and space-time coding in particular and also for cryptography.
Until now the main tool for rank based cryptography was based
on masking Gabidulin codes \cite{GPT91} in different ways and using the McEliece
(or Niederreiter) setting with these codes.
Most cryptosystems based on this idea were broken by using structural attacks which
exploit the particular structure of Gabidulin codes (\cite{O08}, \cite{FL05}, \cite{BL04}, \cite{L11},\cite{G08}). 
A similar situation exists in the Hamming case for which all cryptosystems
based on Reed-Solomon codes have been broken for a similar reason:
Reed-Solomon codes are so structured that they are difficult to mask
and there is always structural information leaking.

Since the introduction of code-based cryptography by McEliece in 1978,
the different cryptosystems proposed in the Hamming distance setting
were based on masking a special
family of decodable codes, like Goppa, Reed-Muller of Reed-Solomon codes.
The strong structure of these codes usually implies a large
public key size.
Now in 1996 and 1997, two lattice-based cryptosystems were proposed independently: 
the NTRU \cite{HPS98} and the GGH \cite{GGH13} cryptosystems which can be seen
as McEliece-type cryptosystems but for the Euclidean distance. Lattice
based cryptography can be seen as code-based cryptography with $q$-ary
(large alphabet) codes
and Euclidean distance rather than the Hamming distance.
Both the NTRU and GGH cryptosystems are based on the same idea: 
knowing a {\it random} basis of small weight vectors
enables one to obtain an efficient decoding algorithm suitable for cryptography.
Moreover, the NTRU cryptosysem (which can be seen as an optimized
case of the GGH cryptosystem \cite{MR09}) introduced for the first
time the idea of using
double-circulant matrices in order to decrease the size of the public key. This
idea was made possible because of the randomness of the small dual basis.
Finally, we remark that for 15 years the NTRU cryptosystem has not been really attacked
on its double-circulant structure, indeed the best-known attacks rely on general LLL (named after their incentors Lenstra–Lenstra–Lovász)
decoding algorithms for lattices.

In a classical cryptographic Hamming context, 
the second author \cite{G05} introduced in 2005 the idea of using
quasi-cyclic codes to decrease the size of the public key. However,
adding quasi-cyclicity to an already structured family of codes introduces
too much structure and the system was broken \cite{OTD10}. 
This idea was then used with other families of quasi-cyclic (or quasi-dyadic) 
structured codes like Goppa quasi-dyadic \cite{MB09} or quasi-cyclic alternant codes \cite{BCGO09}:
these systems lead to much smaller keys, but eventually they were attacked in \cite{FOPT10}
and even though the idea remains valid, the cryptanalysis of \cite{FOPT10}
showed that  the idea of quasi-cyclic or quasi-dyadic structured codes
could not lead to secure public keys of a few thousand bits, but rather to
secure keys of a few tens of thousand bits.

More recently new proposals were made in the spirit of the original
NTRU schemes with Hamming distance, first relying on quasi-cyclic LDPC codes \cite{BCGM07,BBC08,BBC12},
then with MDPC codes in \cite{MTSB12}. The latter family of codes
enables one to obtain
the same type of feature as the NTRU cryptosystem: a compact key
and a security based on decoding with a random dual matrix with small weights.

\medskip

{\bf Contributions of the paper.}
The present paper introduces Low Rank Parity Check codes and is an extended version of \cite{GMRZ13}, 
with important new contributions. 
We propose an improved decoder based on a new approach which allows
for decoding of errors of higher rank weight, namely up to 
$\frac{2}{3}(n-k)$ when
the previous decoding algorithm only performed at rank weight
$\frac{n-k}{2}$. Our codes outperform the classical Gabidulin
code decoder which deals with weights up to $\frac{n-k}{2}$. 
This comes at the expense of probabilistic decoding, but the decoding
error probability can be made arbitrarily small.
For most of our proposed decoding algorithms we give a precise
analysis of the decoding failure probability: for some situations the algorithm is iterative, in which 
case we only give an upper bound on the decoding error probability
together with
simulations which show that these bounds are 
attained in practice.
The new approach can also be used to decrease the decoding error probability, which is especially useful
for cryptography. Finally, we introduce double-circulant {\em ideal} rank codes, which generalize double-circulant
rank codes. We propose different parameter sizes, of the order of 1
kilobit for the key size and cipher text.

\medskip
{\bf Organization of the paper.}
The paper is organized as follows: Section 2 recalls basic facts about
the rank metric and the corresponding difficult problems, Section 3
proves technical results that are used for decoding, 
Section 4 recalls
the definition of LRPC codes and their basic decoding, Section 5 introduces a new general approach for
improved decoding of LRPC codes, finally Section 6 is concerned with cryptography, we give our key exchange
and public key exchange schemes, with a security reduction to general problems and parameters.

\input{rank_metric.tex}

\input{general_decoding_algorithm.tex}

\input{syndrome_space_ext_algo.tex}

\input{applications_to_cryptography.tex}

\newpage
\section{Conclusion}

In this paper, as in the recent MDPC paper \cite{MTSB12}, we generalize the NTRU \cite{HPS98}
approach in a coding context
but with the rank metric. To do so we have introduced a new family of codes, LRPC codes, for
which we propose an efficient decoding algorithm and we have carefully analyzed its failure probability. In order to reduce the public key size, we have chosen to use the special family of Ideal-LRPC codes.

Overall, as it is often the case for rank metric codes, the obtained results compare
well to Hamming distance cryptosystems since the difficulty of known attacks increases.
Moreover, while rank metric cryptosystems have a strong history of broken schemes because
of structural attacks based on recovering the Gabidulin code structure, the cryptosystems
we propose are the first rank-metric based cryptosystems that is not
based on based on Gabidulin codes but on codes with a weak structure
and a strong random component.
%%
%It is also interesting to remark that 
%this type of structure was never really attacked in the case of lattices as it seems
%the case for the MDPC cryptosystem.
%%Gilles: j'ai enlevé cette dernière phrase car elle est totalement
%%incompréhensible
%%
\newpage
\bibliographystyle{plain}
\bibliography{codecrypto}

\appendix
\input{proof_impact_m.tex}

\end{document}

%% file: rank_metric.tex
\section{Background on Rank Metric Codes}
%Let us begin with some general definitions.
\label{sec:RankMetric}
\subsection{General definitions}
\textbf{Notation:}

Let $\Fq$ denote the finite field of $q$ elements where $q$ is the power of a prime and let $\Fqm$ denote the field of $q^m$ elements seen as the extension of degree $m$ of $\Fq$.

$\Fqm$ is also an $\Fq$ vector space of dimension $m$, we denote by capital letters the $\Fq$-subspaces of $\Fqm$ and by lower-case letters the elements of $\Fqm$.

Let $X \subset \Fqm$. We denote by $\vect{X}$ the $\Fq$-subspace generated by the elements of $X$:
\[ \vect{X} = \sum_{x\in X} x\Fq.\]
If $X = \{x_1,\dots, x_n\}$, we simply use the notation $\vect{x_1,\dots, x_n}$.

Vectors are denoted by bold lower-case letters and matrices by bold capital letters (eg $\xv = (x_1,\dots, x_n) \in \Fqm^n$ and $\Mv = (m_{ij})_{\substack{1\leqslant i \leqslant k\\1\leqslant j \leqslant n}} \in \Fqm^{k\times n}$).

If $S$ is a finite set, we denote by $x\rand S$ the fact that $x$ is chosen uniformly at random amongst $S$.

\begin{definition}[Rank metric over $\Fqm^n$] \label{def:RankMetric}
Let $\xv=(x_1,\dots,x_n) \in \Fqm^n$ and let $(b_1,\dots ,b_m) \in \Fqm^m$ be a basis of $\Fqm$ over $\Fq$.
Each coordinate $x_j$ is associated to a vector of $\Fq^m$ in this basis: $x_j = \sum_{i=1}^m m_{ij} b_i$. The $m \times n$ matrix
associated to $\xv$ is given by $\Mv(\xv)=(m_{ij})_{\substack{1
    \leqslant i \leqslant m \\ 1 \leqslant j \leqslant n}}$.

The rank weight $\norme{\xv}$ of $\xv$ is defined as 
\[
\norme{\xv} \eqdef \Rank \Mv(\xv).
\]
This definition does not depend on the choice of the basis.
The associated distance $d(\xv,\yv)$ between elements $\xv$ and $\yv$ in $\Fqm^n$ is defined by 
$d(\xv,\yv)=\norme{\xv-\yv}$.
\end{definition}

\begin{definition}[$\Fqm$-linear code]\label{def:FqmLinearCode}
An $\Fqm$-linear code $\C$ of dimension $k$ and length $n$ is a
subspace of dimension $k$ of $\Fqm^n$ seen as a rank metric space. The
notation $[n,k]_{q^m}$ is used to denote its parameters.

The code $\C$ can be represented by two equivalent ways:
\begin{itemize}
\item by a generator matrix $\Gv \in \Fqm^{k\times n}$. Each rows of $\Gv$ is an element of a basis of $\C$,
\[
\C = \{\xv\Gv, \xv \in \Fqm^k \}.
\]
\item by a parity-check matrix $\Hv \in \Fqm^{(n-k)\times n}$. Each row of $\Hv$ determines a parity-check equation verified by the elements of $\C$:
\[
\C = \{\xv \in \Fqm^n : \Hv\xv^T = \boldsymbol{0} \}.
\]
\end{itemize}
We say that $\Gv$ (respectively $\Hv$) is under systematic form if and only if it is of the form $(\Iv_k|\Av)$ (respectively $(\Bv|\Iv_{n-k})$).
\end{definition}

\begin{definition}[Support of a word]\label{def:support}
Let $\xv = (x_1,\dots, x_n) \in \Fqm^n$. The support $E$ of $\xv$, denoted $\Supp(\xv)$, is the $\Fq$-subspace of $\Fqm$ generated by the coordinates of $\xv$:
\[
E = \langle x_1, \dots, x_n\rangle_{\Fq}
\]
and we have $\dim E = \norme{\xv}$.
\end{definition}

The number of supports of dimension $w$ of $\Fqm$ is given by the Gaussian coefficient 
\[
\cg{w}{m} = \prod_{i=0}^{w-1} \frac{q^m-q^i}{q^w-q^i}.
\]

\subsection{Double circulant and ideal codes}
To describe an $[n,k]_{q^m}$ linear code, we can give a systematic generator matrix or a systematic parity-check matrix. In both cases, the number of bits needed to represent such a matrix is $k(n-k)m\Ceil{\log_2 q}$.  To reduce the size of a representation of a code, we introduce double circulant codes.

First we need to define circulant matrices.
\begin{definition}[Circulant matrix]\label{def:CirculantMatrix}
A square matrix $\Mv$ of size $n\times n$ is  said to be circulant if it is of the form
\[
\Mv = \begin{pmatrix}
m_{0} & m_{1} & \dots & m_{n-1} \\
m_{n-1} & m_{0} & \ddots & m_{n-2} \\
\vdots & \ddots & \ddots & \vdots \\
m_{1} & m_{2} & \dots & m_{0}
\end{pmatrix}.
\]
We denote $\mathcal{M}_n(\Fqm)$ the set of circulant matrices of size $n \times n$ over $\Fqm$.
\end{definition}

\paragraph{Relation between cyclic matrix form and polynomial form} The following proposition states an important property of circulant matrices.
\begin{proposition}\label{prop:MatricesCircPolynomialRing}
$\mathcal{M}_n(\Fqm)$ is an $\Fqm$-algebra isomorphic to $\Fqm[X]/(X^n-1)$, that-is-to-say the set of polynomials with coefficients in $\Fqm$ modulo $X^n-1$. The canonical isomorphism is given by
\begin{IEEEeqnarray}{rCCCc}
\varphi &:	&  \Fqm[X]/(X^n-1) & \longrightarrow &  \mathcal{M}_n(\Fqm) \nonumber \\
			&& \sum_{i=0}^{n-1} m_iX^i & \longmapsto &  \begin{pmatrix}
m_{0} & m_{1} & \dots & m_{n-1} \\
m_{n-1} & m_{0} & \ddots & m_{n-2} \\
\vdots & \ddots & \ddots & \vdots \\
m_{1} & m_{2} & \dots & m_{0} \nonumber
\end{pmatrix}
\end{IEEEeqnarray} 
\end{proposition}

In the following, in order to simplify notation, we will identify the polynomial $G(X) = \sum_{i=0}^{n-1} g_iX^i \in \Fqm[X]$ with the vector $\gv = (g_0,\dots, g_{n-1})\in \Fqm^n$. We will denote $\uv\gv \mod P$ the vector of the coefficients of the polynomial $\left(\sum_{j=0}^{n-1} u_jX^j\right)\left(\sum_{i=0}^{n-1} g_iX^i\right) \mod P$ or simply $\uv\gv$ if there is no ambiguity in the choice of the polynomial $P$.

\begin{definition}[Double circulant codes]\label{def:DoubleCirculantCode}
A $[2n,n]_{q^m}$ linear code $\C$ is said to be double circulant if it
has a generator matrix $\Gv$ of the form $\Gv = (\Av|\Bv)$ where $\Av$
and $\Bv$ are two circulant matrices of size $n$, and the matrix $\Av$ is invertible.
\end{definition}

With the previous notation, we have $\C = \{(\xv\av,\xv\bv),\xv \in
\Fqm^n\}$. Since $\av$ is invertible in $\Fqm[X]/(X^n-1)$, then $\C =
\{(\xv,\xv\gv),\xv \in \Fqm^n\}$ where $\gv = \av^{-1}\bv$. In this
case we say that $\C$ is generated by $\gv\ (\!\!\!\mod X^n-1)$. Thus
we only need $nm\Ceiling{\log_2 q}$ bits to describe a $[2n,n]_{q^m}$
double circulant code.

\medskip

We can generalize double circulant codes by choosing another polynomial $P$ to define the quotient-ring $\Fqm[X]/(P)$. This family of codes was defined in \cite{KLL14}.
% We call this family of codes the ideal code ($\mod P$).
\begin{definition}[Ideal codes]\label{def:IdealCodes}
Let $P(X) \in \Fq[X]$ be a polynomial of degree $n$ and $\gv_1,\gv_2
\in \Fqm^n$. Let $G_1(X) = \sum_{i=0}^{n-1} g_{1i}X^i$ and $G_2(X) =
\sum_{j=0}^{n-1} g_{2j}X^j$ be the polynomials associated respectively to $\gv_1$ and $\gv_2$.

We call the $[2n,n]_{q^m}$ {\em ideal code $\C$ of generator $(\gv_1,\gv_2)$} the code with generator matrix
\[
\Gv = \begin{pmatrix}
G_1(X) \mod P			& \vline &G_2(X) \mod P \\
XG_1(X) \mod P 			& \vline &XG_2(X) \mod P \\
	\vdots				& \vline &\vdots \\
X^{n-1}G_1(X) \mod P 	& \vline &X^{n-1}G_2(X) \mod P 
\end{pmatrix}.
\]
More concisely, we have 
$\C = \{ (\xv\gv_1 \mod P, \xv\gv_2 \mod P), \xv\in \Fqm^n \}$.
We will often omit mentioning the polynomial $P$ if there is no ambiguity.

We usually require $\gv_1$ to be invertible, in which case the code
admits the systematic form, $\C = \{(\xv,\xv\gv), \xv\in \Fqm^n \}$ with $\gv = \gv_1^{-1}\gv_2 \mod P$. 
\end{definition}

%\commAH{peut-etre ici une ou deux phrases pour dire qu'on introduit les codes idéaux pour éviter une attaque spécifique aux codes double circulant cf \ref{subsec:StructuralAttacks} ?}
We need to be careful when we use these notations in the case of parity-check matrices. Indeed, if we have a syndrome $\sigmav = \ev_1\hv_1+\ev_2\hv_2 \mod P$, this equality is equivalent in terms of product matrix-vectors to $(\Hv_1|\Hv_2)(\ev_1|\ev_2)^T = \sigmav^T$ where
\[
\Hv_1 = \begin{pmatrix}
\hv_1 \mod P\\
X\hv_1 \mod P\\
\vdots \\
X^{n-1}\hv_1 \mod P
\end{pmatrix}^T
\text{ and }
\Hv_2 = \begin{pmatrix}
\hv_2 \mod P\\
X\hv_2 \mod P\\
\vdots \\
X^{n-1}\hv_2 \mod P
\end{pmatrix}^T.
\]

Thus, we say that $(\hv_1,\hv_2)$ and $P$ define a parity-check matrix of a code $\C$ if $(\Hv_1^T|\Hv_2^T)$ is a parity-check matrix of $\C$.

\subsection{Difficult problems in rank metric}
In this section, we introduce some difficult problems on which our cryptosystems are based.

\begin{problem}[Rank Syndrome Decoding]\label{prob:RSD}
Given a full-rank matrix $\Hv \in \Fqm^{(n-k)\times n}$, a syndrome $\sigmav$ and a weight $w$, it is hard to find a vector $\xv \in \Fqm^n$ of weight lower than $w$ such that $\Hv\xv^T = \sigmav^T$.
\end{problem}

In \cite{GZ14} it is proven that the Syndrome Decoding problem in the Hamming metric, which is a well-known NP-Hard problem, is probabilistically reduced to the $\RSD$ problem. Moreover, the $\RSD$ problem can be seen as a structured version of the NP-Hard  MinRank problem, indeed the MinRank problem is equivalent to the $\RSD$ problem for $\Fq$-linear codes. 

This problem has an equivalent dual version. Let $\Hv$ be a parity-check matrix of a code $C$ and $\Gv$ be a generator matrix. The $\RSD$ problem is equivalent to finding $\mv \in \Fqm^k$ and $\xv \in \Fqm^n$ such that $\mv\Gv + \xv = \yv$ with $\yv$ of weight $r$ and $\yv$ a pre-image of $\sv$ by $\Hv$. We can now introduce the decisional version of this problem.

\begin{problem}[Decisional Rank Syndrome Decoding]\label{prob:DRSD}
Given a full-rank matrix $\Gv \in \Fqm^{k\times n}$, a message $\mv \in \Fqm^n$ and $e \in \Fqm^n$ of weight $r$, it is hard to distinguish the pair $(\Gv, \mv\Gv + \xv)$ from $(\Gv, \yv)$ with $\yv \rand \Fqm^n$.
\end{problem}

We introduce the $\IRSD$ (Ideal-Rank Syndrome Decoding) problem for ideal codes defined in Definition \ref{def:IdealCodes} as well as an associated problem, $\IRSR$ (Ideal-Rank Support Recovery), and then show that these two problems are equivalent.

\begin{problem}[Ideal-Rank Syndrome Decoding]\label{prob:I-RSD}
Given a vector $\hv \in \Fqm^n$, a polynomial $P\in \Fq[X]$ of degree $n$, a syndrome $\sigmav$ and a weight $w$, it is hard to sample a vector $\xv = (\xv_1,\xv_2) \in \Fqm^{2n}$ of weight lower than $w$ such that $\xv_1 + \xv_2\hv = \sigmav \mod P$.
\end{problem}

Since $\hv$ and $P$ define a systematic parity-check matrix of a
$[2n,n]_{q^m}$ ideal code, the $\IRSD$ problem is a particular case of
the $\RSD$ problem. Although this problem is theoretically easier than
$\RSD$ problem, in practice the best algorithms for solving both these problems are the same.

\begin{problem}[Ideal-Rank Support Recovery]\label{prob:I-RSR}
Given a vector $\hv \in \Fqm^n$, a polynomial $P\in \Fq[X]$ of degree $n$, a syndrome $\sigmav$ and a weight $w$, it is hard to
recover the support $E$ of dimension lower than $w$ such that $\ev_1 + \ev_2\hv = \sigmav \mod P$ where the vectors $\ev_1$ and $\ev_2$ were sampled from $E$.
\end{problem}

\paragraph{Equivalence between the $\IRSD$ and $\IRSR$ problems}

The $\IRSR$ problem is trivially reduced to the $\IRSD$ problem. Indeed to recover the support $E$ of an instance of the $\IRSD$ problem from a solution $\xv$ of the $\IRSD$ problem, we just have to compute the support of $\xv$.

Conversely, the $\IRSD$ problem can also be reduced to the $\IRSR$ problem. Let us suppose we know the support $E$ of a solution of the $\IRSR$ problem for a weight $w$. We want to find $\xv = (\xv_1,\xv_2)$ of weight lower than $w$ such that $\xv_1 + \xv_2\hv = \sigmav \mod P$.

This equation is equivalent to
\begin{align} \begin{pmatrix}
&&\vline && \\
& \Iv_n &\vline &\Hv& \\
&& \vline && 
\end{pmatrix} ( x_{1,0} \dots x_{1,n-1}, x_{2,0} \dots x_{2,n-1})^T = \sigmav^T \label{eqn1}
\end{align}
where 
$ \Hv = \begin{pmatrix}
\hv \\
X\hv  \mod P\\
\vdots \\
X^{n-1}\hv \mod P
\end{pmatrix}^T
$
 and $\xv_1 = ( x_{10} \dots x_{1,n-1}), \xv_2 = (x_{20} \dots x_{2,n-1})$.

Let $(E_1,\dots, E_w)$ be a basis of $E$. We can express the coordinates of $\xv_1$ and $\xv_2$ in this basis:
\[ \forall i \in \{1,2\}, 0 \leqslant j \leqslant n-1, x_{ij} = \sum_{k=1}^w \lambda_{ijk}E_k, \text{ with }\lambda_{ijk} \in \Fq.
\] 
Then we rewrite the equations (\ref{eqn1}) in the new unknowns $\lambda_{ijk}$. We obtain a system of $2nw$ unknowns over $\Fq$ and $n$ equations over $\Fqm$, so $nm$ equations over $\Fq$.

Since $\ev_1 + \ev_2\hv = \sigmav \mod P$, the system has at least one
solution and by construction all the solutions have their support
included in $E$ of dimension $w$, so we can find a solution to the $\IRSD$ problem by solving this system.\\

\section{Some results on the product of two subspaces}
\label{sec:ProductSubspaces}
 
Before introducing the decoding algorithm of LRPC codes, we need to introduce some results on the product 
of two subspaces.

\begin{definition}\label{def:ProductSubspace}
Let $A$ and $B$ be two $\Fq$-subspaces of $\Fqm$:
we call the {\em product space} of $A$ and $B$, and denote it by $AB$, the {\em $\Fq$-linear
  span} of the set of products $\{ab, a \in A, b \in B\}$.
\end{definition}

If $A$ and $B$ have dimensions $\alpha$
and $\beta$, and are generated respectively
by $\{a_1,\cdots,a_{\alpha}\}$ and $\{b_1,\cdots,b_{\beta}\}$,
then the product space $AB$ is obviously
generated by the set $\{a_ib_j, 1 \leqslant i \leqslant \alpha,
1\leqslant j \leqslant \beta \}$ and its dimension is therefore bounded above
by $\alpha \beta$.

A question of interest that concerns us is the probability that the dimension is not
maximal when $\alpha$ and $\beta$ are relatively small.
Let $A$ and $B$ be random $\Fq$-subspaces of $\fqm$ of dimensions $\alpha$
and $\beta$ respectively. We suppose $\alpha\beta <m$ and we
investigate the typical dimension of the subspace $AB$.

We rely on the following lemma:

\begin{lemma}\label{lem:DimensionProductSubspace}
  Let $A'$ and $B$ be two subspaces of $\fqm$ of dimensions $\alpha'$
  and $\beta$ such that $\dim A'B =\alpha'\beta$.
  Let $A = A'+\langle a\rangle$ where $a$ is a uniformly chosen
  random element of $\fqm$. Then
  \[\prob{\dim (AB) < \alpha'\beta +\beta}\leq\frac{q^{\alpha'\beta+\beta}}{q^m}.\]
\end{lemma}

\begin{proof}
  We have $\dim (AB) < \alpha'\beta +\beta$ if and only
  if the subspace $aB$ has a non-zero intersection with $A'B$.
Now,
\begin{align*}
  \prob{\dim (A'B \cap aB) \neq\{0\}} &\leq \sum_{b\in B,
    b\neq 0} \prob{ab\in A'B}\\
  &\leq (|B|-1)\frac{q^{\alpha'\beta}}{q^m}
\end{align*}
since for any fixed $a\neq 0$, we have that $ab$ is uniformly
distributed in $\fqm$. Writing $|B|-1\leq |B|=q^\beta$ we have the result.
\end{proof}

\begin{proposition}\label{prop:DimensionProductSubspace}
  Let $B$ be a fixed subspace and suppose we construct a random subspace $A$ by choosing uniformly at random 
$\alpha$ independent vectors of
$\fqm$ and letting $A$ be the subspace generated by these $\alpha$
random vectors. We have that
$\dim AB =\alpha\beta$ 
with probability at least
$1-\alpha\frac{q^{\alpha\beta}}{q^m}.$
\end{proposition}

\begin{proof}
  Apply the Lemma $\alpha$ times, starting with a random subspace
  $A'\subset A$ of dimension~$1$, and adding a new element to $A'$
  until we obtain $A$.
\end{proof}

In practice, our tests shows that the probability that $\dim AB =rd$ is slightly larger than the probability that $rd$ elements of $\Fqm$ chosen uniformly at random generate a subspace of dimension $rd$. This difference becomes rapidly negligible, even for small value such as $\alpha = 3, \beta = 4, m = 20$.

  Let $B$ be a fixed subspace of $\fqm$ containing $1$ and let $B^2$
  be the subspace generated by all products of two, possibly equal, elements of
  $B$. Let $\beta_2=\dim B^2$. Let $A$ be a random
  subspace of $\fqm$ of dimension $\alpha$. By Proposition \ref{prop:DimensionProductSubspace} we have that
  $\dim (AB^2) =\alpha\beta_2$
  with probability at least
  $1-\alpha\frac{q^{\alpha\beta_2}}{q^m}.$\\
 {\bf Remark:} we have $\beta_2\leqslant \beta(\beta +1)/2$.

  \begin{lemma}\label{lem:xB}
    Suppose $\dim (AB^2) =\alpha\beta_2$. Let $e\in
    AB$ with $e\not\in A$. Suppose $eB\subset 
    AB$. Then there exists $x\in B$, $x\not\in\fq$, such that $xB\subset B$.
  \end{lemma}
  
  \begin{proof}
    Let $(a_i)$ be a basis of $A$. We have
   $$e=\sum_i\lambda_ia_ib_i$$
   with $\lambda_i \in \fq$ for all $i$ and $b_j\notin\fq$ and $\lambda_j\neq
   0$ for some $j$, otherwise $e\in A$ contrary to our assumption. Let
   $b$ be any element of $B$. By our hypothesis we have $eb\in 
    AB$, meaning there exist $\mu_j \in \Fq$ such that
  $$\sum_i\lambda_ia_ib_ib = \sum_j\mu_j a_j b_j'$$
  with $b_i'\in B$. Now the maximality of the dimension of $AB^2$ implies that
  $$\lambda_ja_jb_jb = \mu_ja_jb_j'$$
  from which we deduce $b_jb\in B$. Since this holds for arbitrary
  $b\in B$, we have  $b_jB\subset B$.
  \end{proof}

  \begin{proposition}\label{prop:DimensionIntersectionAllSi}
    Suppose $m$ is prime. Let $A$ and $B$ be random subspaces of
    dimensions $\alpha$ and $\beta$ respectively. Let $(b_i)$ be a
    basis of $B$ and let $S=AB$. Then with
    probability at least
    $1-\alpha\frac{q^{\alpha\beta(\beta+1)/2}}{q^m}$
    we have that 
    $\bigcap_{i}b_i^{-1}S = A.$
  \end{proposition}
  
    \begin{proof}
    If not, there exists a subspace $E\supsetneq A$, such that 
    $EB = AB$. By the remark before
    Lemma~\ref{lem:xB} we have that with probability at least
    $$1-\alpha\frac{q^{\alpha\beta(\beta+1)/2}}{q^m}$$
    the conditions of Lemma~\ref{lem:xB} hold. But then there is
    $x\not\in\fq$ such that $xB\subset B$. But this implies that
    $\fq(x)B\subset B$. But $m$ prime implies that there is no
    intermediate field between $\fq$ and $\fqm$, hence
    $\fqm\subset B$, a contradiction.
  \end{proof}
  
 \paragraph{A better bound}

Our goal is to show that with a large probability, when $A$ and $B$
are randomly chosen with sufficiently small dimension, then with
probability close to $1$ we have that :
$$\bigcap_{b\in B,b\neq 0} b^{-1} AB = A.$$
Without loss of generality we can suppose that $1\in B$. We shall show
that for a random $b\in B$, we have that $b\neq 0$ and
$$AB \cap b^{-1} AB  = A$$
with probability close to $1$.

When $b$ is a random element of $\fqm$ we will
abuse notation somewhat by letting $b^{-1}$ denote
the inverse of $b$ when $b\neq 0$ and by letting it equal $0$ when $b=0$.
With this convention in mind we have:

\begin{lemma}\label{lem:B1}
  Let $B_1$ be a subspace of dimension $\beta_1$. Let $b$ be a
  uniformly distributed random element of $\fqm$. Then the probability
  that $b\in B_1+b^{-1}B_1$ is at most:
  $$\frac{2q^{2\beta_1}}{q^m}.$$
\end{lemma}

\begin{proof}
  This event can only happen if $b$ is a root of an equation of the
  form
  $$x^2-b_1x-b_1'=0.$$
  There are at most $|B_1|^2=q^{2\beta_1}$ such equations, and each
  one of them has at most two roots.
\end{proof}

Let $B_1$ be any vector space containing $1$ and of dimension $\beta_1$.
Let $b$ be a random element uniformly distributed in $\fqm$
and set $B=B_1+\langle b\rangle.$ Denote $\beta =\dim B=\beta_1 +1$
(with probability $1-q^{\beta_1}/q^m$).
Since $b^{-1}$ is also uniformly
distributed, so is $b^{-1}b_1$ for any $b_1\in B_1$, $b_1\neq 0$.
Therefore, for any such $b_1$, the probability that $b^{-1}b_1\in B_1$
equals $|B_1|/q^m$. By the union bound, the union of these events,
over all $b_1\in B_1\setminus\{0\}$,  has
probability at most $(|B_1|-1)|B_1|/q^m$. We have therefore that
either $b=0$ or $B_1\cap b^{-1}B_1 \neq\{0\}$ with probability
at most 
$$(|B_1|-1)\frac{|B_1|}{q^m}\leq \frac{q^{2\beta_1}}{q^m}.$$
Therefore with probability at least
$$1-\frac{q^{2\beta_1}}{q^m}$$
we have that $b\neq 0$ and
$$\dim (B_1+b^{-1}B_1) = 2\beta_1.$$
Now applying Lemma~\ref{lem:B1} we obtain:
\begin{lemma}
  With probability at least
  $$1-\frac{3q^{2\beta_1}}{q^m}$$
we have that
  $$\dim (B+b^{-1}B) = 2\beta -1.$$
\end{lemma}

\begin{proposition}\label{prop:DimensionIntersection}
  Let $B$ be a subspace of dimension $\beta$
  containing $1$ such that $\dim B+b^{-1}B=2\beta-1$ for some $b\in B$.
  Let $A$ be a randomly chosen subspace of dimension $\alpha$.
  With probability at least
  $1-\alpha\frac{q^{\alpha(2\beta-1)}}{q^m}$
  we have that
  $AB \cap b^{-1} AB = A$
\end{proposition}

\begin{proof}
  By Proposition~\ref{prop:DimensionProductSubspace} we have that with probability at least
  $$1-\alpha\frac{q^{\alpha(2\beta-1)}}{q^m}$$
  $$\dim [A(B+b^{-1}B)]= \alpha(2\beta-1) = 2\alpha\beta
  -\alpha.$$
  On the other hand, we have that
  \begin{align*}
\dim  A(B+b^{-1}B) &= \dim AB +
  \dim b^{-1}AB - \dim( AB\cap b^{-1}AB)\\
  &= 2\alpha\beta - \dim( AB\cap ABb^{-1})
  \end{align*}
hence 
  $$\dim(AB\cap b^{-1}AB) = \alpha.$$
But this proves the result since $A\subset AB\cap b^{-1}AB$ and $\dim A=\alpha$.
  \end{proof}

\begin{remark} It is interesting to remark that in practice the probabilities
 given in Proposition \ref{prop:DimensionIntersectionAllSi} and \ref{prop:DimensionIntersection} decrease much  faster
to $0$. Indeed when the degree of the extension $m$ increases by one (for $m$
greater than $\alpha\beta$), if we make the reasonable assumption that each subspace of the form $b_i^{-1}S$ behaves as a random subspace which contains $A$, the probability that one intersection  $b_i^{-1}S \cap b_j^{-1}S \ne A$ is divided by $q$, hence the probability that $\bigcap_{i}b_i^{-1}S \ne A$ is divided by a factor at least $q^{\beta-1}$. 
%if we consider f_i^-1EF as radom subspace which contains E, proba E \subset ... divided by q when m++. With \beta-1 intercestion /q^{\beat-1}. In fact, much more
This means that in practice the previous upper
bound is rather weak, and that one can consider that as soon
as $m$ is greater than $\alpha\beta$ by $8$ or more (and increasing) the probability is far below $2^{-30}$.
This will be the case when we choose parameters in the last section.
\end{remark}

%% file: general_decoding_algorithm.tex
\section{LRPC codes and their basic decoding} \label{sec:basicAlgo}

\subsection{Low Rank Parity Check codes}

LRPC codes were introduced in \cite{GMRZ13}. They are good candidates
for the McEliece cryptosystem of because the have a weak algebraic structure.
\begin{definition}[LRPC codes]\label{def:LRPC}
Let $\Hv = (h_{ij})_{\substack{1\leqslant i \leqslant n-k\\ 1
    \leqslant j \leqslant n}} \in \Fqm^{(n-k)\times n}$ be a full-rank matrix such that its coefficients generate an $\Fq$-subspace $F$ of small dimension $d$:
\[
F = \langle h_{ij}\rangle_{\Fq}.
\]
Let $\C$ be the code with parity-check matrix $\Hv$. By definition, $\C$ is an $[n,k]_{q^m}$ LRPC code of dual weight $d$.
Such a matrix $\Hv$ is called a homogeneous matrix of weight $d$ and support $F$.
\end{definition}

We can now define ideal LRPC codes similarly to our definition of ideal codes. %To keep the structure of the parity-check matrix $\Hv$ of support $F$, we cannot choose any polynomial $P$. Indeed if the first row of $\Hv$ is $(\hv_1 | \hv_2) \in F^n \times F^n$, then its $i^{th}$ row is $(X^{i-1}\hv_1(X) \mod P | X^{i-1}\hv_2(X) \mod P)$ and we need to have $X^{i-1}\hv_1(X) \mod P \in F^n$ and $X^{i-1}\hv_1(X) \mod P \in F^n$. That implies $P \in \Fq[X]$.

\begin{definition}[Ideal LRPC codes]\label{def:I-LRPC}
Let $F$ be an $\Fq$-subspace of dimension $d$ of $\Fqm$, let
$(\hv_1,\hv_2)$ be two vectors of $\Fqm^n$ of support $F$ and let $P
\in \Fq[X]$ be a polynomial of degree $n$.
% Let $\Xv$ and $\Yv$ be the matrix of the product by $\xv$ and $\yv$ in $\Fqm[X]/(P)$.
Let
\[
\Hv_1 = \begin{pmatrix}
\hv_1 \\
X\hv_1 \mod P \\
\vdots \\
X^{n-1}\hv_1 \mod P
\end{pmatrix}^T
\text{ and }
\Hv_2 = \begin{pmatrix}
\yv \\
X\hv_2 \mod P \\
\vdots \\
X^{n-1}\hv_2 \mod P
\end{pmatrix}^T.
\]
When the matrix $\Hv = (\Hv_1|\Hv_2)$ has rank $n$ over $\Fqm$,
the code $\C$ with parity check matrix $\Hv$ is called an ideal LRPC code of type $[2n,n]_{q^m}$.
\end{definition}

As we can see, since $P \in \Fq[X]$, the support of $X^i\hv_1$ is still $F$ for all $1 \leqslant i \leqslant n-1$ hence the necessity to choose $P$ with coefficients in the base field $\Fq$ to keep the LRPC structure of the ideal code.

\subsection{A basic decoding algorithm}
The general idea of the algorithm is to use the fact that we know a parity-check matrix with a small weight $d$ that is to say that the $\Fq$-subspace $F$ of $\Fqm$ spanned by the coordinates of this matrix has a small dimension of $d$, hence the
subspace $S= \vect{s_1,\dots,s_{n_k}}$ generated by the coordinates of syndrome 
enables one to recover the whole product space $P= EF$ of the support of the error and
of the known small basis of $H$. Knowing the space $P$ and
the space $F$ enables one to recover $E$.
Then, knowing the support $E$ of the error $\ev$, it is easy to recover 
the exact value of each coordinate by solving 
a linear system. This approach is very similar to the classical decoding procedure
of BCH codes for instance, where one recovers the error-locator polynomial, which gives
the support of the error, and then the value of the error coordinates.

Consider an $[n,k]_{q^m}$ LRPC code $\C$ of low weight $d$, 
with generator matrix $\Gv$ and dual $(n-k) \times n$ matrix $\Hv$, such that all
the coordinates $h_{ij}$ of $\Hv$ belong to a space $F$ of dimension $d$ with basis
$\{f_1,\dots,f_d\}$.

Suppose the received word to be $\yv = \xv \Gv + \ev$ 
where $\ev \in \Fqm^n$ 
is the error vector of rank~$r$. Let $E = \Supp(\ev)$ and 
let $\{e'_1, \ldots, e'_r\}$ be a basis of $E$.

We have the basic decoding algorithm described in Figure~\ref{fig:Decoding algorithm of LRPC codes}, which
has a probability of failure that we will consider in subsection \ref{sub:prob_failure}.

%\newpage

\begin{figure}[!ht]
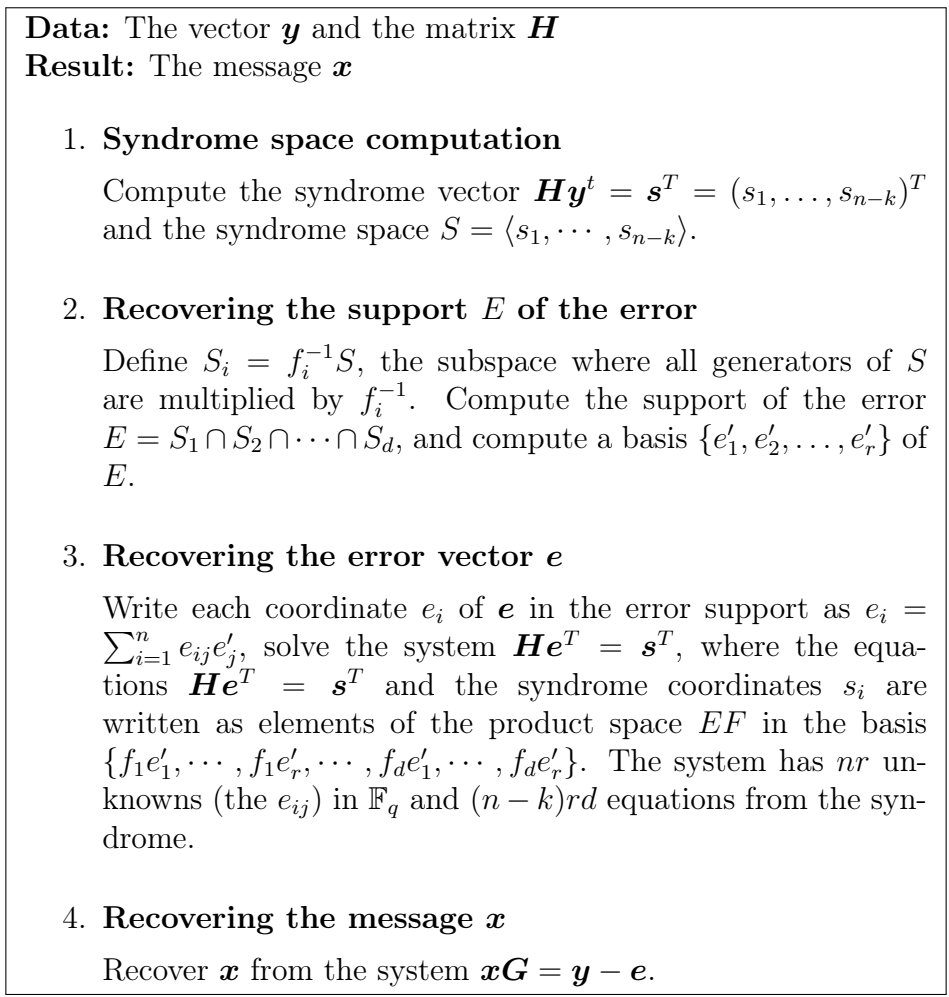

\centering
\fbox{
  \begin{minipage}{12cm}
  
  {\bf Data:} The vector $\yv$ and the matrix $\Hv$
  
  {\bf Result:} The message $\xv$
  
\begin{enumerate}
\item {\bf Syndrome space computation}

Compute the syndrome vector $\Hv\yv^t=\sv^T = (s_1,\dots,s_{n-k})^T$ and the syndrome space 
$S=\vect{s_1,\cdots,s_{n-k}}$.

\medskip

\item {\bf Recovering the support $E$ of the error}

Define $S_i=f_i^{-1}S$, the subspace where all generators of $S$ are multiplied by 
$f_i^{-1}$.
Compute the support of the error $E = S_1 \cap S_2 \cap \cdots \cap S_d$,
and compute a basis $\{e'_1,e'_2,\dots,e'_r\}$ of $E$.

\medskip

\item {\bf Recovering the error vector $\ev$}

Write each coordinate $e_i$ of $\ev$ in the error support as $e_i=\sum_{i=1}^ne_{ij}e'_j$,
solve the system $\Hv\ev^T=\sv^T$, where the equations $\Hv\ev^T=\sv^T$ and the syndrome coordinates $s_i$
are written as elements of the product space $\EF$ in the basis 
$\{f_1e'_1,\cdots,f_1e'_r,\cdots,f_de'_1,\cdots,f_de'_r\}$. The system has $nr$ unknowns (the $e_{ij}$) in $\Fq$ and $(n-k)rd$ equations from the syndrome.

\medskip

\item {\bf Recovering the message $\xv$}

Recover $\xv$ from the system $\xv\Gv=\yv-\ev$.

 \end{enumerate}
  \end{minipage}
} \caption{A basic decoding algorithm for LRPC codes}
\label{fig:Decoding algorithm of LRPC codes}
\end{figure}

%.B_1^{- 1} S = \text{(d\'ef)} \tmop{Vect}
%(B_1^{- 1} s_1, \ldots, B_1^{- 1} s_{n - m})$. $S_1$ est clairement un espace
%vectoriel.

\subsection{Correctness of the algorithm}

We prove the correctness of the algorithm in the ideal case when 
$\dim \EF =rd$, $\dim S=rd$ and  
$\dim S_1 \cap S_2 \cap \cdots \cap S_d =r$.
We will see in the next subsection that this is the general case.

{\bf step 1}: the first step of the algorithm is obvious.

\medskip

{\bf step 2}: now we want to prove that 
$ E \subset S_1 \cap S_2 \cap \cdots \cap S_d$.
We define $S_i=f_i^{-1}S=\{f_i^{-1}x, x \in S\}$. Since by hypothesis
the space
$S$ is {\it exactly} the product space $\EF$, we have
$f_ie'_j \in~S, \forall 1 \le j \le r$, hence $e'_j \in S_i$ for all $i,j$. Therefore $E \subset S_i$,
 hence  $E \subset S_1 \cap S_2 \cap \cdots \cap S_d$. 
By hypothesis $\dim S_1 \cap S_2 \cap \cdots \cap S_d= \dim E$ hence
$E=S_1 \cap S_2 \cap \cdots \cap S_d$.

\medskip

{\bf step 3}: once the support $E$ of the error $\ev$ is known, one can write
\[e_i=\sum_{1 \le j \le r} e_{ij}e'_j, \text{ with }e_{ij} \in \Fq\]
and solve the linear system $\Hv\ev^T=\sv^T$ in the $nr$ unknowns $e_{ij}$.
The system has $nr$ unknowns in $\Fq$ and $(n-k)rd$ equations in $\Fq$
 coming from the $n-k$ syndrome equations in $\EF$. The parameter
$d$ is chosen such that $d\ge \frac{n}{n-k}$.

\subsection{Probability of failure} \label{sub:prob_failure}

We now consider the different possibilities of failure,
there are three cases to consider. The case $\dim \EF = rd$ corresponds 
to proposition \ref{prop:DimensionProductSubspace}, the case $E = S_1 \cap S_2 \cap \cdots \cap S_d$ 
corresponds to proposition \ref{prop:DimensionIntersectionAllSi}. In both cases the probability 
can be made exponentially small depending on the parameters, especially when in practice
the upper bound given are really large compared to experimental results.

The last case is the case $\dim S=rd$.  We have the following easy proposition: 
\begin{prop} \label{prop:failure prob}
The probability that the $n-k$ syndrome coordinates do not generate 
the product space $\EF$ is less than $q^{rd-(n-k)}$.
\end{prop}
\begin{proof} By construction all $s_i$ belong to the product space $\EF$
and since the error is taken randomly and the matrix $\Hv$ is
full-rank, the $s_i$ can be seen as random
elements of $\EF$.
Now if one considers a set of $(n-k)$ random elements in a space of dimension $rd$
(with $n-k \ge rd$) 
the probability that this set does not generate the whole space is
equal to the probability that a random 
$[rd,n-k]$ matrix
over $\Fq$ is not full-rank, which is not more than $q^{(n-k)-rd}$
according to the lemma below.
\end{proof}

\begin{lemma}
  Let $a,x$ be integers and let $\Av$ be a uniform random $a\times (a+x)$ matrix over $\Fq$. The
  probability that $\Av$ has rank less than $a$ is bounded from above by $q^{-x}$.
\end{lemma}

\begin{proof}
  Proceed by induction on $a$. When $a=1$, the probability that $\Av$ is
  not full-rank is the probability that its unique random row equals
  the zero row, which is equal to $q^{-(x+1)}<q^{-x}$. Let now $a>1$ and
  suppose the result proved for matrices with $a-1$ rows. The
  probability of $\Av$ not being full-rank is, by the law of total probability, not more than 
  $$P_1+P_2$$
  where $P_1$ is the probability that $\Av$ is not full-rank,
  conditional on its first $a-1$ rows being full-rank, and where $P_2$
  is the probability that the first $a-1$ rows of $\Av$ are not full-rank.
  we have $P_1\leq q^{a-1}/q^{a+x}=1/q^{x+1}$ and $P_2\leq q^{-x-1}$
  by the induction hypothesis. Hence the probability of $\Av$ not
  being full-rank is at most $2q^{-x-1} \leq q^{-x}$.
\end{proof}

The previous discussion shows that depending on the parameters the probability
of failure of the previous algorithm can be made arbitrarily small and
that the main probability we have to consider in fact is the probability given
by Proposition \ref{prop:failure prob}, which is not an upper bound but what happens in practice.

\subsection{Computational cost of decoding}

The most costly step of the algorithm are step 2) and step 3).
The cost of step 2) is the cost of the intersection of vector spaces
which has cost $4 r^2 d^2 m$ operations in the base field
(this operation can also be done in a very elegant way with $q$-polynomials \cite{O33}).
Now the cost of step 3) is the cost of solving the system $\Hv\ev^T=\sv^T$ when the support
$E$ of the error is known. If one proceeds naively there are 
$nr$ unknowns (the $e_{ij}$) and the cost of matrix inversion is $n^3r^3$,
but one can use a precomputation to perform this step in $n^2r$ operations.
%now one can use the formal decoding matrix $\Dv_{\Hv}$ of the previous section and simply 
%recover the $(e_{ij})$ by multiplying by $\Dv_{\Hv}$ the $nr$ positions 
%(written in the product basis of $\langle E.F \rangle$) of $s_1,..,s_{n-k}$ associated 
%to the matrix $\Dv_{\Hv}$ of definition 6. Therefore the cost of the inversion
%becomes only the cost of a matrix multiplication: $n^2r^2$.
%Remark that the matrix $\Dv_{\Hv}$ can be precomputed and stocked or even 
%reconstructed column by column from random hash values - in that case
%one fixes $\Dv_{\Hv}$ and one derives $H$.

To solve the linear system in $n^2r$ operations, we introduce the matrix $\Av_{\Hv}^r$, which corresponds to the matrix $\Hv$ "unfolded" over $\Fq$ and represents the system $\Hv\ev^T=\sv^T$.

We want to solve the following system :

\begin{IEEEeqnarray}{rccl}
\left(
\begin{tabular}{cccc}
$h_{11}$ & $h_{12}$ & $\cdots$ & $h_{1n}$\\
$h_{21}$ & $\ddots$ & & $\vdots$\\
$\vdots$ & & $\ddots$ & $\vdots$\\
$h_{(n-k)1}$ & $\cdots$ & $\cdots$ & $h_{(n-k)n}$
\end{tabular}
\right) & \left(
\begin{tabular}{c}
$e_1$\\
$e_2$\\
$\vdots$\\
$e_n$
\end{tabular}
\right) & = & \left(
\begin{tabular}{c}
$s_1$\\
$s_2$\\
$\vdots$\\
$s_{n-k}$
\end{tabular}
\right)
\end{IEEEeqnarray}

Suppose $\dim \EF$ is exactly $rd$. Then we can unfold each $s_i$ in the basis $(f_ie'_j)_{\substack{1 \leqslant i \leqslant d\\ 1 \leqslant j \leqslant r}}$. Now if we unfold each $h_{ij}$ in the basis $(f_u)_{1 \leqslant u \leqslant d}$ : $h_{ij} = \sum\limits_{u=1}^d h_{iju}f_u, h_{iju}\in \Fq$ and each $e_j$ in the basis $(e'_v)_{1 \leqslant j \leqslant r}$ : $e_j = \sum\limits_{v=1}^r e_{jv}e'_v, e_{jv}\in \Fq$, we get the following relation: 
\[
s_i = \sum\limits_{j=1}^n h_{ij}e'_j = \sum\limits_{j=1}^n\sum\limits_{u=1}^d\sum\limits_{v=1}^r h_{iju}e_{jv} f_u e'_v.
\]
 Each of these relations can be viewed as a linear system with $nr$ unknowns (the $e_{ij}$) and $rd$ equations from the syndrome in $\Fq$ :

\begin{IEEEeqnarray}{rccl}
\left(
\begin{tabular}{cccccccccc}
$h_{111}$ & $0$ & $0$ & $h_{121}$ & $0$ & $0$ & & $h_{1n1}$ & $0$ & $0$\\
$0$ & $\ddots$ & $0$ & $0$ & $\ddots$ & $0$ & $\cdots$ & $0$ & $\ddots$ & $0$\\
$0$ & $0$ & $h_{111}$ & $0$ & $0$ & $h_{121}$ & & $0$ & $0$ & $h_{1n1}$\\

$h_{112}$ & $0$ & $0$ & $h_{122}$ & $0$ & $0$ & & $h_{1n2}$ & $0$ & $0$\\
$0$ & $\ddots$ & $0$ & $0$ & $\ddots$ & $0$ & $\cdots$ & $0$ & $\ddots$ & $0$\\
$0$ & $0$ & $h_{112}$ & $0$ & $0$ & $h_{122}$ & & $0$ & $0$ & $h_{1n2}$\\

 & $\vdots$ & & & $\vdots$ & & & & $\vdots$ & \\ 

$h_{11d}$ & $0$ & $0$ & $h_{12d}$ & $0$ & $0$ & & $h_{1nd}$ & $0$ & $0$\\
$0$ & $\ddots$ & $0$ & $0$ & $\ddots$ & $0$ & $\cdots$ & $0$ & $\ddots$ & $0$\\
$0$ & $0$ & $h_{11d}$ & $0$ & $0$ & $h_{12d}$ & & $0$ & $0$ & $h_{1nd}$\\
\end{tabular}
\right) & \left(
\begin{tabular}{c}
$e_{11}$\\
$e_{21}$\\
$\vdots$\\
$e_{1r}$\\
$e_{21}$\\
$\vdots$\\
$e_{2r}$\\
$\vdots$\\
$e_{n1}$\\
$\vdots$\\
$e_{nr}$
\end{tabular}
\right) & = & \left(
\begin{tabular}{c}
\begin{tabular}{c}
$s_{111}$\\
$\vdots$\\
$s_{1d1}$\\
$s_{112}$\\
$\vdots$\\
$s_{1d2}$\\
$\vdots$\\
$s_{11r}$\\
$\vdots$\\
$s_{1dr}$
\end{tabular}
\end{tabular}
\right)
\label{systemMatrix}
\end{IEEEeqnarray}

Where $s_{1uv} = \sum\limits_j h_{1ju}e_{jv}$ (the syndrome coordinates unfolded in the basis of $\EF$). Each $h_{1ju}$ is expanded into an $r \times r$ diagonal matrix whose coefficients are equal to $h_{1ju}$. By repeating this process for each $s_i$, we obtain a linear system with $nr$ unknowns and $(n-k)rd$ equations in $\Fq$. We call the matrix representing this system $\Av_{\Hv}^r$. The first line of $\Av_{\Hv}^r$, for example, represents the impact of the error vector $\ev$ on the first line of $\Hv$ relative to the element ${F_1E_1}$ of the basis of $\EF$.

Denote $\Av_{\Hv}$ an invertible submatrix of $\Av_{\Hv}^r$ and $\Dv_{\Hv} = \Av_{\Hv}^{-1}$. Then solving the linear system consists only of multiplying $nr$ syndrome coordinates by an $nr \times nr$ $\Dv_{\Hv}$ instead of inverting it.

Since $\Av_{\Hv}$ is a $nr \times nr$ matrix, it can be chosen such that it only contains $r \times r$ diagonal matrices. $\Av_{\Hv}$ and $\Dv_{\Hv}$ can therefore be stored in memory as $n^2$ elements of $\Fq$, and can be inverted in $n^3$ operations in $\Fq$ instead of $n^3r^3$ operations. Using this representation, we can see that the multiplication $\Dv_{\Hv}.\sv$ takes only $n^2r$ multiplications in $\Fq$ since each row of the condensed matrix contributes to $r$ error coordinates, hence the complexity. This complexity can also be attained with the full $nr \times nr$ matrix by taking into account the fact that it is composed of multiple $r \times r$ diagonal matrices.

%% file: syndrome_space_ext_algo.tex
\section{Improved decoding : syndrome space expansion algorithm} \label{sec:syndrome space expansion algorithm}

In general, the basic decoding algorithm presented in Section \ref{sec:basicAlgo} does not work when the syndrome space $S = \langle s_1, ..., s_{n-k} \rangle$ is different from $\EF$. In this section, we present an algorithm that can be used between steps 1) and 2) of the basic decoding algorithm to recover $\EF$ in certain cases when $\dim S < \dim \EF$. We denote by $c$ the codimension of $S$ in $\EF$, that is to say $c = \dim \EF - \dim S$.

In this section, we will always suppose that $\dim \EF = rd$ for two reasons:
\begin{itemize}
\item first, according to Proposition \ref{prop:DimensionProductSubspace}, the probability $p_1$ that $\dim \EF < rd$ can easily be exponentially low by increasing $m$. We can also decrease the probability $p_2$ that $S \ne \EF$ by increasing $n$, however for cryptographic applications of LRPC, the parameters of the code are such that $p_1$ is negligible with respect to $p_2$.
\item secondly, the case $\dim \EF = rd$ is the worst case for the decoding since the probability that $S = \EF$ increases when the dimension of $\EF$ decreases. Thus, we can analyze the theoretical probability of failure in the worst case scenario. All the tests we have made show that in practice the probability of failure is smaller when $\dim S < rd$, which confirms our analysis.
\end{itemize}

\subsection{General idea}

This algorithm's aim is to recover vectors of the product space $\EF$ which
do not belong to the syndrome space $S$ by using the structure of the product
space $\EF = \langle f_1e_1, f_2e_1, ..., f_de_r \rangle$ and the fact
that we know a basis $\langle f_1, ..., f_d \rangle$ of $F$. 
%We call a function that increases the dimension of $S$ an ``expansion function''.
We use a function which take on input the subspace $S$ and outputs a subspace $S'$ such that $S \varsubsetneq S' \subset \EF$ with a very high probability, depending on the parameters of the code. We call such a function an ``expansion function''.

Here are two examples of expansion functions:

\begin{enumerate}
\item $f_{decode}(S) =  (S + f_iS_j) \cap (S + f_kS_l)$ is used to decode errors of larger weight.
  %allow increase error weights
\item $f_{prob}(S) = S + FS_{ij} $, where $S_{ij} = S_i \cap S_j$, is used to reduce the decoding failure rate.
\end{enumerate}

We use these expansion functions to describe an iterative algorithm [\ref{fig:Syndrome space expansion algorithm}]. Detailed algorithms for each function are given in the next subsections.

\begin{figure}[h!]
\centering
\fbox{
  \begin{minipage}{12cm}
  
  {\bf Data:} The syndrome space $S$ and the vector space $F$
  
  {\bf Result:} The expanded syndrome space $S'$, which may be $\EF$, or failure
\begin{enumerate}
\item {\bf Expand $S$ with the expansion function}

Use the expansion function for all possible parameters:\\
$S \leftarrow f_{expand}(S)$

\medskip

\item {\bf Iterative step}

If $\dim S = rd$, then return $S$. If $\dim S < rd$, go back to step~1. If $\dim S$ does not increase during step 1, return failure.

 \end{enumerate}
  \end{minipage}
} \caption{The syndrome space expansion algorithm}
\label{fig:Syndrome space expansion algorithm}
\end{figure}

\begin{prop} \label{prop:dimSi}
$\dim (S_i \cap E) \geqslant r-c$, with $c = \dim \EF - \dim S$.
\end{prop}

\begin{proof}
We have $\dim S_i = \dim S = rd -c$ and
\begin{IEEEeqnarray}{rCl}
\dim (E+S_i) & = & \dim f_i(E+S_i) = \dim (f_i E+S) \nonumber \\
&\leqslant &\dim \EF =rd
\end{IEEEeqnarray}
since $f_i E$ and $S$ are subspaces of $\EF$.
Hence:
\begin{IEEEeqnarray}{rCCCCCl}
\dim (S_i \cap E)	&	 = 		& \dim E	& +	& \dim S_i	& -	& \dim (E +S_i)\nonumber \\
				& \geqslant	& r			& +	& rd-c		& -	& rd \nonumber \\
				& \geqslant	& r-c. \nonumber
\end{IEEEeqnarray} 
\end{proof}

As we will see in the next subsections, each expansion function has
advantages and drawbacks, and the choice of which function to use
depends on the parameters. In particular, $m$ needs to be high enough
for these functions to work. The function $f_{decode}$ needs $m \ge 3rd-2$ and allows to decode
errors of weight up to $\frac{2}{3}(n-k)$ whereas the function $f_{prob}$ only needs $m \ge 2rd-r$. In this case
we can only decode errors of weight up to $\frac{(n-k)}{2}$ but by strongly reducing the decoding failure probability. In particular, this latter case corresponds to parameters used for cryptography, where $rd \approx (n-k)=\frac{n}{2}$ and $m \approx n$.

\subsection{Increasing the weight of the decoded error vectors}

In this subsection we explain in details the function $f_{decode}$:
\[
f_{decode} : S \mapsto (S + f_iS_j) \cap (S + f_kS_l).
\]
The idea is
that each subspace of the form $(S + f_iS_j)$ has a chance to contain
vectors $x$ such that $x \in \EF$ and $x \notin S$, since $\EF \cap f_iS_j \ne \{0\}$. If we find some new vectors of $\EF$ but not the whole product space, we can
reuse values of $(i, j, k, l)$ already tested because they may lead us
to recover some new vectors thanks to the one we have added. Hence the algorithm functions in an iterative way. If during a whole iteration, the dimension of $S$ does not increase, then the algorithm will fail.

\begin{algorithm}[H]
\KwIn{The syndrome space $S$ and the vector space $F$}
\KwOut{The expanded syndrome space $S'$, which may be $\EF$, or failure}
\Begin{
	\While{\textit{true}} {
		$tmp \leftarrow \dim S$ \;
		\For{Every $(i, j, k, l)$ such that $i \neq j$, $k \neq l$ and $(i, j) \neq (k, l)$}{
			$S \leftarrow (S + f_i.S_j) \cap (S + f_k.S_l)$ \;
		}
		\If{$\dim S = rd$}{\Return{$S$} \;}
		\If{$\dim S = tmp$}{\Return{failure} \;}
	}
}

\caption{Syndrome space expansion using $f_{decode}$}
\label{fig:expansion1}
\end{algorithm}

\paragraph{Maximum weight of the decoded errors.}

Without using the syndrome space expansion algorithm, the decoding
algorithm for LRPC codes can decode up to $r = \frac{n-k}{d}$
errors. We can use $f_{decode}$ to deal with error vectors of increased weight.

\begin{theorem}
For $d=2$, under the assumption that $\dim (S_i \cap E) = r-c$ for $1 \leqslant i \leqslant 2$, $f_{decode}$ allows to decode errors of weight up to $r \leqslant \Floor{\frac{2}{3}(n-k)}$ with probability $\approx 1 - \frac{q^{3r-2(n-k)}}{q-1}$ for the general case and errors of weight up to  $r=\frac{2}{3}(n-k)$ with probability $\approx 0.29$ for the case $q=2$.
\end{theorem}

\proof
%In the case $d=2$, when computing $\EF + f_i.f_j^{-1}.\EF$, the expected dimension of this vector space is $rd + r(d-1) = 3r$, since $r$ vectors of $f_i.f_j^{-1}.\EF$ will already be in $\EF$ and the other ones are random. If we get smaller vector spaces when computing $S + f_i.f_j^{-1}.S$ then it means that we can not find the whole space $\EF$, which leads to a decoding failure.
%
%Since the dimension of $S$ is at most $n-k$, the dimension of $S + f_i.f_j^{-1}.S$ is at most $2(n-k)$, which leads to the following inequality : $3r \leqslant 2(n-k)$, hence the result.
%
%In the $d=2$ case, the only expansion that can be done is $S \leftarrow (S + f_1.f_2^{-1}.S) \cap (S + f_2.f_1^{-1}.S)$, hence there is no iterative effect. As $(S + f_1.f_2^{-1}.S) = f_1f_2^{-1}(S + f_2.f_1^{-1}.S)$, the dimension of the two terms is the same. Considering $S$ and $f_1.f_2^{-1}.S$ are independent, the probability of failure is the probability that two subspaces of $\EF + f_i.f_j^{-1}.\EF$ of dimension $n-k$ do not generate the whole space, which is roughly $0.28$ for $q=2$ and close to $q^{-1}$ for $q>2$.
$d=2 \implies \dim \EF = 2r$.\\
Thus, the dimension of $\EF+ f_1f_2^{-1}\EF$ is at most $3r$ for $f_1 E \subset \EF \cap f_1f_2^{-1}\EF$. The only expansion that can be done is $S \leftarrow (S + f_1f_2^{-1}S) \cap (S + f_2f_1^{-1}S)$, hence there is no iterative effect.

 Since $(S + f_1f_2^{-1}S) = f_1f_2^{-1}(S + f_2f_1^{-1}S)$, the dimension of these two vectors spaces is the same. If  $ \dim (S + f_1 f_2^{-1}S) < 3r$, then it is impossible to find the whole space $\EF$, which leads to a decoding failure. Since the dimension of $S$ is at most $n-k$, the dimension of $S + f_if_j^{-1}S$ is at most $2(n-k)$, which implies $3r \leqslant 2(n-k)$. Hence $r \leqslant \Floor{\frac{2}{3}(n-k)}$.
 
 Considering that $S$ and $f_1f_2^{-1}S$ are independent, the probability of success is the probability that the sum of two subspaces of dimension $n-k$ of $\EF + f_1f_2^{-1}\EF$ generates the whole space. So we have:
\begin{IEEEeqnarray}{rCl}
p_{success} & =& \Prob (S + f_1 f_2^{-1}S = \EF + f_1f_2^{-1}\EF)  \nonumber \\
& = & \Prob (\dim (S + f_1 f_2^{-1}S) = 3r) \nonumber \\
& =& \Prob (2(n-k) - \dim S \cap (f_1 f_2^{-1}S) = 3r) \nonumber \\
& =& \Prob (\dim S \cap (f_1 f_2^{-1}S) = 2(n-k) - 3r). \nonumber
\end{IEEEeqnarray}
The formula for computing this probability can be found in \cite{H14} (Proposition 3.2). We obtain:
\begin{IEEEeqnarray}{rCl}
p_{success} & =& \frac{\cg{2(n-k)-3r}{n-k}q^{(3r-(n-k))^2}}{\cg{n-k}{3r}} \nonumber \\
 & =& \frac{\cg{3r-(n-k)}{n-k}q^{(3r-(n-k))^2}}{\cg{3r-(n-k)}{3r}}. \label{eqn:ExactFormulaProbSucces|d=2}
\end{IEEEeqnarray}
By applying the series expansion for $P_{a,b}(n)$  (Appendix \ref{app:impact_m},page \pageref{DvptLimiteProba}) with $n=3r, a=b=3r-(n-k)$, we obtain:
\begin{IEEEeqnarray}{rCl}
p_{success} & \approx & 1-\frac{q^{-3r}(q^{3r-(n-k)}-1)^2}{q-1}  \nonumber \\
 & \approx & 1- \frac{q^{3r-2(n-k)}}{q-1}. \nonumber
\end{IEEEeqnarray}

In the case $q=2$ and $3r=2(n-k)$, this approximation is incorrect. According to Equation \eqref{eqn:ExactFormulaProbSucces|d=2},
\begin{IEEEeqnarray}{rCl}
p_{success} & = & \frac{\cgdeux{n-k}{n-k} 2^{(n-k)^2}}{\cgdeux{n-k}{3r}} \nonumber \\
& = & \frac{2^{(n-k)^2}}{\cgdeux{n-k}{2(n-k)}} \nonumber \\
%& = & \frac{2^{(n-k)^2}{\prod_{i=0}^{n-k-1} \frac{2^{2(n-k)}-2^i}{2^{n-k}-2^i}}  \nonumber \\
%& = & \prod_{i=0}^{n-k-1} \frac{2^{2(n-k)}-2^{n-k+i}}{2^{n-k}-2^i} \nonumber \\
& \longrightarrow &  0.29 \text{ when } (n-k)\rightarrow \infty \text{ (the convergence is very fast)}.\nonumber
\end{IEEEeqnarray}
\qed

In figure \ref{simud2}, we present simulation results comparing the observed probability of failure and $q^{-1}$.

\begin{figure}
        \centering
        \includegraphics[scale=0.7]{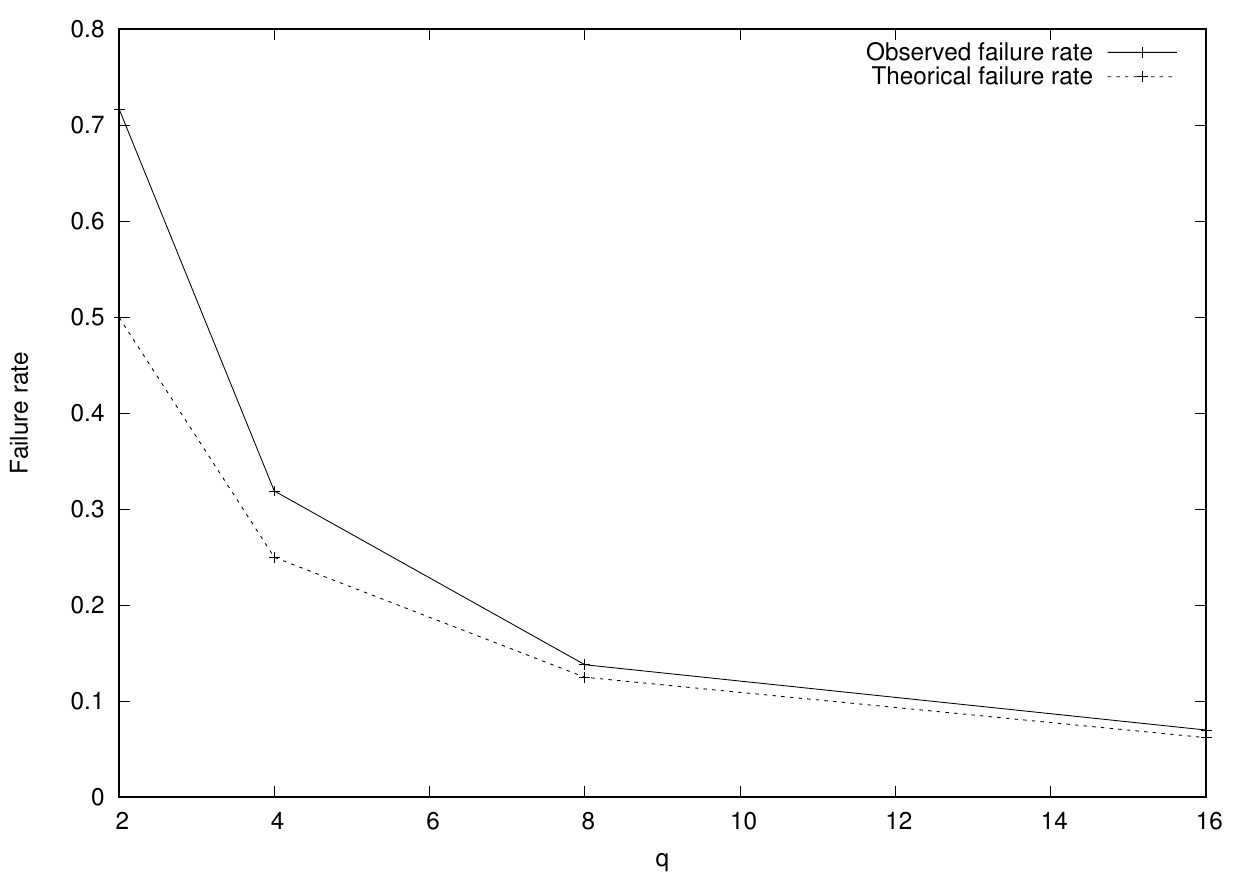}
        \caption{Results for $d=2$, $n=30$, $k=15$, $r=10$. We can observe the success rate of $\approx 0.29$ for $q=2$.}
        \label{simud2}
\end{figure}

\begin{proposition}
Under the assumption that $\dim (S_i \cap E) = r-c$ for $1 \leqslant i \leqslant d$, $f_{decode}$ allows to decode errors up to $r \leqslant \frac{n-k}{d-1}$ for $d>2$, using the fact that the algorithm functions in an iterative way.
\end{proposition}

\proof
For $d>2$, the algorithm may recover $\EF$ as long as $c<r$ (i.e there are vectors of $E$ in $S_i$) using the iterative effect, so the algorithm does not have to recover the whole product space at a time : if it recovers one new vector then it can be added to $S$ and used with different values of $(i, j, k, l)$ (even those that have already been used) to produce more vectors, hence the only condition for that method to work is $c<r$ in order to have vectors of $E$ in $S_i$.

If the $n-k$ syndrome vectors are independent, then $c$ is $rd - (n-k)$, which gives the result:

\begin{IEEEeqnarray}{rCl}
c < r & \Leftrightarrow & rd-(n-k)<r \nonumber \\
      & \Leftrightarrow & r(d-1)<n-k \nonumber \\
      & \Leftrightarrow & r<\frac{n-k}{d-1}. \nonumber
\end{IEEEeqnarray}
\qed

Although giving a theorical probability of failure is difficult in this case since there is an iterative effect, simulations show that in practice we attain this bound.

\paragraph{Minimal value of $m$}

\begin{proposition}
$f_{decode}$ recovers new vectors of $\EF$ only if $m > 3rd - 2$
\end{proposition}

\proof
To produce vectors of $\EF$, we compute vector spaces of the form $f_if_j^{-1}S$. If we suppose that $f_j^{-1}S$ contains exactly $r-c$ independent vectors of $E$, which is the worst case, then we get $r-c$ vectors of $\EF$ and $r(d-1) + c$ random vectors in $S'_{ij} = f_if_j^{-1}S$. 

By adding this vector space to $S$, the dimension cannot exceed $rd + r(d-1) + c = 2rd - r + c$ i.e the dimension of $\EF$ along with the random vectors in $S'_{ij}$.

The intersection between $S'_{ij}$ and $S'_{kl}$ needs to remove all the random vectors in order to recover vectors that are $\in \EF$, whence:

\begin{IEEEeqnarray}{rCl}
\dim(S'_{ij}) + \dim(S'_{kl}) \leqslant m + \dim(S'_{ij} \cap S'_{kl}) & \Leftrightarrow & 2(2rd - r + c) \leqslant m + rd \nonumber \\
                                                                    & \Leftrightarrow & m \geqslant 4rd - 2r + 2c - rd \nonumber \\
                                                                    & \Leftrightarrow & m \geqslant 3rd - 2r + 2c. \nonumber
\end{IEEEeqnarray}

Since $c$ can not be higher than $r-1$ (there would be no vectors of $E$ in $f_j^{-1}S$), the inequality becomes $m \geqslant 3rd - 2$.
\qed

The value of $m$ needed by this expansion function is pretty high and does not fit the parameters used for cryptography. However, it works for $d=2$, which is not the case of $f_{prob}$ studied in the next subsection.

\subsection{Reducing the decryption failure rate}

\subsubsection{Description of the algorithm}

In this subsection we will study $f_{prob}$ in a more detailed
way. 
\[
f_{prob} : S \mapsto S + FS_{ij}.
\]
This function fits the parameters used for cryptography and is
used to reduce the decoding failure rate. It uses the fact that
$S_{ij} = S_i \cap S_j$ contains at least $r-2c$ independent vectors of $E$ (see Proposition \ref{prop:dim Sij cap E}) and is a subspace of $E$ with a very high probability when $m$ is large enough. Thus, the subspace $FS_{ij}$ has a chance to contain a vector $x \in \EF$ and $x \notin S$.\\
\begin{algorithm}[H]
\KwIn{The syndrome space $S$ and the vector space $F$}
\KwOut{The expanded syndrome space $S'$, which may be $\EF$, or failure}

\Begin{
	\While{\textit{true}} {
		$tmp \leftarrow \dim S$ \;
		\For{Every $(i, j)$ such that $i \neq j$}{
			$S \leftarrow S + FS_{ij}$ \;
		}
		\If{$\dim S = rd$}{\Return{$S$} \;}
		\If{$\dim S = tmp$}{\Return{failure} \;}
	}
}

\caption{Syndrome space expansion using $f_{prob}$}
\label{fig:expansion2}
\end{algorithm}

\begin{proposition}\label{prop:dim Sij cap E}
$\dim (S_{ij} \cap E) \geqslant r-2c$ for all $i,j \in [1..d]$, with $c = \dim \EF - \dim S$.
\end{proposition}

\proof
Using Proposition \ref{prop:dimSi}, a simple computation gives us the result:
\begin{IEEEeqnarray}{rCl}
\dim (S_{ij}\cap E)	&=& \dim (S_i\cap S_j \cap E)  \nonumber \\
					&=&  \dim (S_i \cap E) + \dim (S_j \cap E)	- \dim \big( (S_i \cap E)+(S_j \cap E)\big) \nonumber \\
					&\geqslant & 2(r-c) 					- r \nonumber \\
					&\geqslant & r-2c. \nonumber
\end{IEEEeqnarray}
\qed

\begin{theorem}
Let $\Prob(c=l)$ be the probability that $\dim \EF - \dim S = l$ and
let $\Prob_{c=l}(failure)$ be the probability that the syndrome space expansion algorithm does not recover $\EF$ when $c=l$.
Using the syndrome space expansion algorithm with $f_{prob}$, the probability of failure is :

$$\prob{c \geqslant \frac{r}{2}} + \sum\limits_{l=1}^{\Floor{\frac{r}{2}} - 1} \Prob(c=l) \times \Prob_{c=l}(failure).$$

\end{theorem}

\proof
The main source of decoding failures for the basic decoding algorithm is the fact that the syndrome coordinates do not generate the whole product space $\EF$, so we need to compute the global probability of not recovering $\EF$ using $f_{prob}$.

When $c \geqslant \frac{r}{2}$, in general $S_{ij}$ is empty so the syndrome space expansion algorithm will not recover $\EF$, which leads to a failure.

When $c < \frac{r}{2}$, the probability that the syndrome space expansion algorithm recovers $\EF$ is $\Prob_{c=l}(failure)$. We multiply this probability by the probability of $c$ being equal to $l$ to get the overall probability of not recovering $\EF$.
\qed

\begin{proposition} \label{prop:prob_c_eq_l}
$\Prob(c=l) \simeq q^{-l(n-k-rd+l)}$
\end{proposition}

\proof

The probability that the $n-k$ syndrome coordinates generate a subspace of codimension $l$ of $\EF$ is the probability that a random $(n-k) \times rd$ matrix in $\Fq$ is of rank $rd - l$. From \cite{L06}, we have that $S_t$, the number of $m \times n$ matrix of rank $t$, is :

$$S_t = \prod_{j=0}^{t-1} \frac{(q^n - q^j)(q^m - q^j)}{(q^t - q^j)}.$$

We can approximate this formula by ignoring $q^j$, which leads to :

$$S_t \simeq \prod_{j=0}^{t-1} \frac{(q^n)(q^m)}{(q^t)} = \prod_{j=0}^{t-1} q^{n + m - t} = q^{t(n+m-t)}.$$

From that, we deduce the number of $(n-k) \times rd$ matrices of rank $rd - l$ :

$$q^{(rd-l)(n-k+rd-rd+l)} = q^{(rd-l)(n-k+l)}.$$

By dividing this quantity by the total number of $(n-k) \times rd$ matrices, we get the probability that a random matrix has rank $rd -l$ :

$$\frac{q^{(rd-l)(n-k+l)}}{q^{(n-k)rd}} = q^{(rd-l)(n-k+l) - (n-k)rd} = q^{lrd - l(n-k) - l^2} = q^{-l(n-k-rd+l)}.$$

Hence the result. \qed

\subsubsection{Case of codimension 1}

Before giving the probability of success of Algorithm \ref{fig:expansion2}, we need the following lemma:

\begin{lemma} \label{dimSicapE}
The dimension of $S_i \cap E$ is :
\begin{itemize}
\item $r$ with probability $p_r = \frac{\cg{r}{rd-1}}{\cg{r}{rd}} \approx q^{-r}$.
\item $r-1$ with probability $1-p_r \approx 1 - q^{-r}$.
\end{itemize}
\end{lemma}
\begin{proof}
  We have:
$\dim S_i \cap E = r \iff E \subset S_i \iff f_iE \subset S$.\\
The space $f_iE$ is a subspace of dimension $r$ of $\EF$ and $S$ is a random subspace of dimension $rd-1$ of $\EF$, hence
\[
\Prob(\dim (S_i \cap E) = r) = \frac{\cg{r}{rd-1}}{\cg{r}{rd}} \approx q^{-r}
\]
which proves the first point.
According to Proposition \ref{prop:dimSi}, $\dim (S_i \cap E) \geqslant r-1$ which concludes the proof of this lemma.
\end{proof}

We can now give the probability of success of the decoding algorithm in the case $\dim S = \dim \EF -1$.

\begin{theorem} \label{failure_exp2}
Under the assumption $S_{ij} \subset E$ and $\dim S
= \dim \EF -1$, the syndrome space expansion algorithm
using $f_{prob}$ recovers $\EF$ with probability at least
$1-q^{(2-r)(d-2)}$.
\end{theorem}

\proof 

To find an upper bound on the probability of failure, we are going to
focus on the case where $\dim(S_{ij} \cap E) = r-2$. Indeed, Lemma
\ref{dimSicapE} shows that the typical dimension of $S_i \cap E$ is
$r-1$. Furthermore, when there exists $i$ such that $\dim(S_i \cap E)
= r$, the algorithm has access to more vectors of $E$ and so this
leads to a smaller probability of failure. The same thing goes for the
case where there exists $\dim(S_{ij}) = r-1$ instead of $r-2$ which is
the expected size, so we can consider that $\dim(S_{ij}) = r-2$ since
this will lead to an upper bound of the probability of failure.
% Pas très bien dit %%%%%Juste pas lisible !

We are now going to study the dimension of $S_{ij} + S_{jk} + S_{ik}$. There are two possibilities:

\begin{itemize}
\item If $S_{ij} = S_{jk}$, then necessarily we have $S_{ik} = S_{ij} = S_{jk}$, because the $r-2$ vectors of $S_{ij} + S_{jk}$ are both in $S_i$ and $S_k$. This happens with probability $q^{2-r}$ : the probability that two subspaces of dimension $r-2$ ($S_{ij}$ and $S_{jk}$) of a vector space of dimension $r-1$ ($S_j \cap E$) are equal.
\item If $S_{ij} \neq S_{jk}$, then we have $S_{ik} + S_{ij} + S_{jk} = E$ : we know that $S_{ik} + S_{ij} = S_{i} \cap E$, and that $S_{jk} \not\subset S_i$ (otherwise we would be in the first case). Hence $S_{ik} + S_{ij} + S_{jk} = E$.
\end{itemize}

If there exists a set $\{S_{i_1j_1}, \dots, S_{i_nj_n}\}$ such that their direct sum equals $E$, then the algorithm will succeed since it will have added exactly $\EF$ to $S$ at one point. If such a set does not exist, then it means that every $S_{ij}$ are equal to each other. In particular, every vector of $E$ the algorithm as access to is $\in S_1, \dots, \in S_d$, which means that when we multiply this vector by $F$, the resulting vectors will already be in $S$ and we will obtain no information, causing the algorithm to fail.

The probability that all the $S_{ij}$ are equal is the probability of $S_{12} = S_{23} = \dots = S_{d-1,d}$. Since $d-2$ equalities are needed, the probability of this happening is $q^{(2-r)(d-2)}$, which gives the probability of failure. \qed

\paragraph{Impact of $m$ on decoding.}

If the value of $m$ is not high enough, then $S_{ij}$ will contain not only vectors of $E$ but also random vectors. When this happens, the $S_{ij}$ cannot be used by the expansion algorithm because it would add random vectors to $S$.

\begin{prop} \label{prop:impact_m}
Let $E$ be a subspace of $\Fqm$ of dimension $r$. Let $F_1$ and $F_2$ be two subspaces of $\Fqm$ such that \begin{itemize}
\item $\dim F_i = d_i$
\item $\dim (E \cap F_i) = r -c_i$
\end{itemize}
Then $\dim (F_1 \cap  F_2) \geqslant r-c_1-c_2$ and we have
\[
\prob{(F_1\cap F_2) \subset E} \geqslant \frac{\cg{d_2-r+c_2}{m-d_1-c_1}}{\cg{d_2-r+c_2}{m-r}}q^{(d_1-r+c_1)(d_2-r+c_2)}.
\]
If $m \geqslant d_1+d_2$, we have $\prob{(F_1\cap F_2) \subset E} \geqslant 1-\frac{q^{-m+r}(q^{d_1-r+c_1}-1)(q^{d_2-r+c_2}-1)}{q-1} + \OO{q^{-2m}}$ when $m \rightarrow \infty$.
\end{prop}

The proof of this proposition is given appendix \ref{app:impact_m}.

\begin{cor} \label{cor:formule_m}
Applying Proposition \ref{prop:impact_m} with $F_i = S_i$, $\dim S_i = rd-1$ and $c_i = 1$, we obtain $\prob{(S_i\cap S_j) \subset E} \approx 1 - \frac{q^{-m+2rd-r}}{q-1}$.
\end{cor}

From corollary \ref{cor:formule_m}, we can see that adding $1$ to the value of $m$ leads to a factor $q^{-1}$ in the probability of finding random vectors in $S_{ij}$, thus this probability can be made arbitrarily small.

To detect that $S_{ij}$ contains random vectors, we can keep track of the dimension of $S$ during the expansion. Indeed, if $S_{ij}$ contains random vectors, adding $FS_{ij}$ to $S$ will lead to a dimension greater than $rd$. In this case we just need to discard the vectors we just added and continue with the next $S_{ij}$.

\paragraph{Comments on the small values of $q$}

When using small values of $q$, the cases where $\dim S_i \cap E$ or
$\dim S_{ij} > r-2$ are not so rare. Simulation for $d=3, r=3$ seems
to show that the probability $q^{(2-r)(d-2)}$ is accurate when $q$ is
high enough, but is a pessimistic bound when $q$ is small, especially
when $q=2$. Simulation results are shown in figure \ref{fig:simud3r3}. 

\begin{figure}
\begin{center}
\includegraphics[scale=0.75]{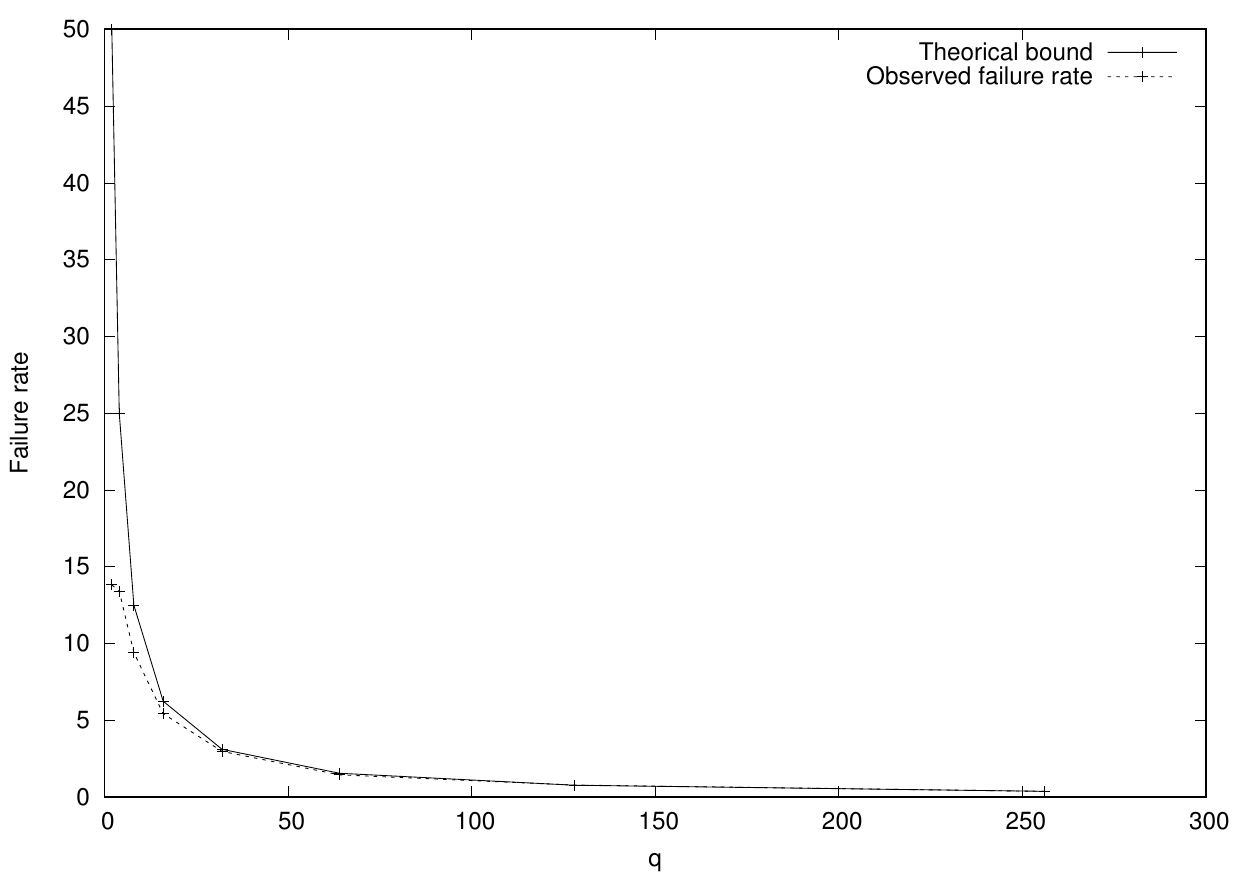}
\end{center}
\caption{Simulation results for $c=1$, $d=3, r=3$}
\label{fig:simud3r3}
\end{figure}

\begin{proposition}
For $q=2$, the syndrome space expansion algorithm using $f_{prob}$ recovers $\EF$ with probability at least $1-q^{(1-r)(d-2)}$.
\end{proposition}

\proof When $q=2$, the probability that $S_{ij} = S_{jk}$ used in the proof above is not $q^{2-r}$ but $q^{1-r}$. Indeed, from Proposition \ref{prop:impact_m}, we get that the probability of $S_{ij}$ and $S_{jk}$ being equal is approximately $\frac{q^{-1}}{q-1}$ which is equal to $q^{-1}$ when $q=2$. \qed

This leads to a failure probability of $q^{(1-r)(d-2)}$. This is still not as good as what happens in practice, since it does not take into account the fact that the algorithm often has more than $r-2$ vectors of $E$ in each $S_{ij}$ to work with, but is a closer estimate in this special case.

\subsubsection{Case of codimension $\geqslant$ 2}

In the case $c \geqslant 2$, we have no theoretical analysis, but we can do simulations for particular parameters to estimate the probability of failure. In practice we obtain good estimations of the probability of failure. For example, for $q=2$, $c=2$, $r=5$, $d=6$ and $m=80$, we obtain a probability of failure of $2^{-14}$.

\subsubsection{Upper bound on the decryption failure rate}

\begin{itemize}
\item The codimension of $S$ in $\EF$ is 1, but the algorithm fails to recover $\EF$.
\item The codimension of $S$ in $\EF$ is greater than 1. Since it is difficult to give a theorical bound on the probability, we will consider that decoding fails in that case:
\end{itemize}

%To summarize : the probability of the codimension of $S$ in $\EF$ being 1 is $q^{(n-k)-rd}$ (proposition \label{prop:failure prob}), and the algorithm fails with probability $q^{(2-r)(d-2)}$ (theorem \ref{failure_exp2}) when this is the case. The algorithm can also fail because of the codimension being greater than 1, which happens with probability $\simeq q^{-2(n-k-rd+2)}$ (proposition \ref{prop:prob_c_eq_l}), and we consider that the algorithm fails in that case. That gives us the following upper bound of the decryption failure rate :

$$q^{(2-r)(d-2)} \times q^{(n-k)-rd} + q^{-2(n-k-rd+2)}.$$

\subsection{Tradeoff between expansion functions}

In the previous subsections we studied two expansion functions in detail:

\begin{itemize}
\item $f_{decode}$ can decode higher values of $r$ but needs a higher $m$ ($m \ge 3rd-r$) 
\item $f_{prob}$ fits values of $m$ ($m \ge 2rd-r$) used in cryptography, but is less effective since it decodes up to $c < \frac{r}{2}$ instead of $c < r$
\end{itemize}

Depending on the parameters, we can build new expansion functions: in general, the higher $m$ is, the better the expansion function will be. For example, if $m$ is really high, we can use this function: $S \leftarrow (S + f_iS_j + f_kS_l) \cap (S + f_aS_b + f_cS_d)$ which is basically an improvement of $f_{decode}$ since it has a higher chance of recovering vectors of $EF$ by adding two $f_iS_j$ at a time.

If $m$ is lower than the one needed for $f_{prob}$, one can use the function $S \leftarrow S + F(S_{ij} \cap S_k)$ which requires a lower $m$ but can only decode with $c < \frac{r}{3}$.

\subsection{Rank Support Recovery algorithm}

Algorithm \ref{fig:Decoding algorithm of LRPC codes} computes the error vector from the syndrome equations $\Hv\ev^T=\sv^T$. However, for cryptographic applications, we may just recover the support of the error and use it as the shared secret since the $I-RSR$ and $I-RSD$ problems are equivalent. We use the Rank Support Recovery (RSR) algorithm which is a shortened version of the decoding algorithm \ref{fig:Decoding algorithm of LRPC codes} which does not compute the coordinates of the error, and uses $f_{prob}$ to reduce the decoding failure rate.

\begin{algorithm}
\label{algo:RSR}
\KwIn{A parity-check matrix $\Hv \in \Fqm^{(n-k)\times n}$ of support $F$ of dimension $d$ of an LRPC code, a syndrome $\sv \in \Fqm^{n-k}$, an integer $r$.}
\KwOut{A subspace $E$ of dimension $\leqslant r$ such that it exists $\ev \in \Fqm^n$ of support $E$ with $\Hv\ev^T=\sv^T$.}
\Begin{
Compute the syndrome space $S=\vect{s_1,\cdots,s_{n-k}}$ \;
Use the syndrome space expansion algorithm described in Algorithm \ref{fig:expansion2} \;
Define $S_i=f_i^{-1}\EF$ then
compute the support of the error $E = S_1 \cap S_2 \cap \cdots \cap S_d$\;
\Return $E$
}
\caption{the Rank Support Recovery algorithm}
\end{algorithm}

Algorithm \ref{fig:expansion2} allows to reduce the decryption
failure rate by a factor $q^{(2-r)(d-2)}$. However, this algorithm
takes a variable number of iterations to terminate, which may lead to
big  variations in the global execution time. In particular, if one
can determine whether the syndrome space expansion algorithm was
needed or not, the execution time may leak some information. To
prevent that, we present a variation of the syndrome space recovery
algorithm [\ref{algo:crypto_expansion}] that performs a fixed number of intersections and still fits the proof of Theorem~\ref{failure_exp2}.

\begin{algorithm}
\label{algo:crypto_expansion}
\KwIn{The syndrome space $S$ and the vector space $F$}
\KwOut{The expanded syndrome space $S'$, which may be $\EF$}
\Begin{
\For{$i$ from 1 to $d-1$} {
	Precompute $S_{i, i+1}$ \;
}
\For{$i$ from 1 to $d-2$} {
	$tmp \leftarrow S + F(S_{i, i+1} \oplus S_{i+1, i+2} \oplus S_{i, i+2})$ \;
	\If{$\dim tmp \leqslant rd$} {
		$S \leftarrow tmp$ \;
		\tcc{$tmp$ is used to check that $m$ has no impact, has in proposition \ref{prop:impact_m}}
	}
}
\Return{$S$}
}
\caption{Syndrome space expansion algorithm used in cryptography}
\end{algorithm}

\subsection{Complexity of the decoding algorithm}

Algorithms \ref{fig:expansion1} and \ref{fig:expansion2} are iterative algorithms, hence their exact complexity is hard to study precisely. We can however give an upper bound of their complexity, by noticing that the maximum number of iterations before the algorithms finishes is at most $c$: indeed, $c$ must decrease at each iteration, otherwise the algorithms will stop returning $failure$.

\begin{proposition} \label{prop:comp_fdecode}
The number of operations in $\Fq$ for the syndrome space expansion algorithm using $f_{decode}$ in algorithm \ref{fig:expansion1} is bounded above by:

$$ c \times d(d-1)\times(d(d-1)-1) \times 16 r^2 d^2 m$$
\end{proposition}

\begin{proof}
As said before, the complexity of this algorithm can be bounded by $c$ times the cost of an iteration. When using $f_{decode}$, the most costly operation is the computation of the intersection. Intersecting two vector spaces of dimension $2rd$ costs $16 r^2 d^2 m$ operations in $\Fq$, and $d(d-1)\times(d(d-1)-1)$ different values are possible for $(i, j, k, l)$ at each iteration, hence the result.
\end{proof}

\begin{proposition} \label{prop:comp_fprob}
The number of operations in $\Fq$ for the syndrome space expansion algorithm using $f_{prob}$ in algorithm \ref{fig:expansion2} is bounded above by :

$$ c \times d(d-1) \times 4 r^2 d^2 m$$
\end{proposition}

\begin{proof}
The proof is the same as for proposition \ref{prop:comp_fdecode}, except that the most costly step of each iteration is the computation of the $S_{ij}$. Intersecting two vector spaces of dimension $rd$ costs $4 r^2 d^2 m$ operations in $\Fq$, and there are $d(d-1)$ $S_{ij}$ to compute at each iteration, hence the result.
\end{proof}

In order to allow constant time implementations, the syndrome space expansion algorithm used in the cryptographic context computes a fixed number of intersections, so its complexity can be derived easily:

\begin{proposition}
The number of operations in $\Fq$ for the syndrome space expansion algorithm used
in the cryptographic context in algorithm \ref{algo:crypto_expansion} is:

$$ ((d-1)+(d-2)) \times 4r^2d^2m$$
\end{proposition}

\begin{proof}
As for Proposition~\ref{prop:comp_fprob}, the most costly step of this algorithm is the computation of every $S_{ij}$. Each intersection costs $4 r^2 d^2 m$, and the structure of the algorithm allows to reduce the number of intersections to $(d-1)+(d-2)$, hence the result.
\end{proof}

In practice the algorithm for cryptographic applications is very efficient, and
for the iterative case, the upper bound is scarcely met since the algorithm
usually stops before doing the maximum number of possible iterations.

%% file: applications_to_cryptography.tex
\section{Applications to cryptography}\label{sec:applicationsToCrypto}
LRPC codes are good candidates for use in the McEliece cryptosystem
\cite{M78}. As a reminder, the McEliece cryptosystem is a public key
encryption scheme based on coding theory. Basically, the secret key is
the description of a code $\C$ that we can efficiently decode and the
public key is a scrambled description of this code. To encrypt a
message $M$, one computes a codeword $c_M \in \C$ and one adds to it
an error $e$. The ciphertext is $C = c_M+e$. If we know the ``good''
description, we can decode $C$ and recover the message whereas an
attacker has to either apply generic decoding algorithms for arbitrary
codes, which amounts to trying to solve an NP-hard problem, or recover the hidden structure of $\C$.

Ideal LRPC codes are easy to hide since we only need to reveal its
systematic parity-check matrix. Due to their weak algebraic structure,
it is hard to recover the structure of an LRPC code from its systematic form. We can now introduce a new problem on which the security of our cryptosystem is based: 

\begin{problem}[Ideal LRPC codes indistinguishability]\label{prob:IndLRPC} Given a polynomial $P \in \Fq[X]$ of degree $n$ and a vector $\hv \in \Fqm^n$, it is hard to distinguish whether the ideal code $\C$ with the parity-check matrix generated by $\hv$ and $P$ is a random ideal code or whether it is  an ideal LRPC code of weight $d$.

In other words, it is hard to distinguish whether $\hv$ was sampled uniformly at random or as $\xv^{-1}\yv \mod P$ where the vectors $\xv$ and $\yv$ have the same support of small dimension $d$.
\end{problem}

%The ideal LRPC codes are particularly interesting if we choose an irreducible polynomial for $P$. In this case we counter a structural attack against double circulant LRPC which can be found in \cite{HT15}.
In \cite{HT15}, a structural attack against DC-LRPC codes is made by using the divisors of $X^n-1$ in $\Fq[X]$. That is why the ideal LRPC codes are particularly interesting if we choose an irreducible polynomial for $P$. In this case we utterly counter this structural attack.
%We discuss the difficulty of this problem in the subsection \ref{subsubsec:StructuralAttacks}.

The most efficient known attack against this problem consists of
searching for a codeword of weight $d$ in the dual $\C^\perp$ of the
LRPC code. Indeed, $\hv = \xv^{-1}\yv \mod P \iff \yv + \xv\hv = 0
\mod P$ so $(\xv|\yv) \in \C^\perp$. The small weight codeword
research algorithms are the same as for the algorithm solving the $\RSD$ problem, so we consider the algorithm of \cite{AGHT17}. However this algorithm is more efficient for LRPC codes than for random codes. Indeed, since an LRPC code admits a parity-check matrix $\Hv$ of small weight $d$, any $\Fq$-linear combinations of the rows of $\Hv$ is a codeword of weight at most $d$ of $\C^\perp$. In the case of Ideal-LRPC codes, the complexity of the algorithm is divided by $q^n$. Thus the complexity of the attack  is $\OO{n^3m^3q^{d\Ceiling{\frac{m}{2}}-m-n}}$.

We can also consider the algebraic attacks using the Groebner bases \cite{LP06}.
The advantage of these attacks is that they are independent of the
size of $q$. They mainly depend on the number of unknowns with respect
to the number of equations. However, in the case $q=2$ the number of
unknowns is generally too high for the algorithms by Groebner basis to
be more efficient than combinatorial attacks. We have chosen our
parameters in such a way so that the best attacks are combinatorial: the expected complexity of the algorithms by Groebner basis is based on the article \cite{BFP09}.

\subsection{An IND-CPA KEM based on the rank metric}
\label{subsec:KEM}

\subsubsection*{Definition of a KEM}
\label{subsubsec:defKEM}
A Key-Encapsulation scheme KEM = $(\KeyGen,\Encap,\Decap)$ is a triple of probabilistic algorithms together with a key space $\mathcal{K}$. The key generation algorithm $\KeyGen$ generates a pair of public and secret keys $(\pk,\sk)$. The encapsulation algorithm $\Encap$ uses the public key $\pk$ to produce an encapsulation $c$, and a key $K \in \mathcal{K}$. Finally $\Decap$ using the secret key $\sk$ and an encapsulation $c$, recovers the key $K \in \mathcal{K}$ or fails and return $\bot$.

We define $\INDCPA$-security of KEM formally via the following experiment, where $\Encap_0$ returns a valid key pair $c^*,K^*$, while $\Encap_1$ return a valid $c^*$ and a random $K^*$.\\[2mm]
\begin{tabular}{rl}
\begin{minipage}{0.4\textwidth}
\textit{Indistinguishability under Chosen Plaintext Attack}: This notion states that an adversary should not be able to efficiently guess which key is encapsulated.
\end{minipage}
	&
\fbox{
\begin{minipage}{0.4\textwidth}
$\Exp_{\E,\A}^{\ind-b}(\seck)$
\begin{enumerate}
\item $\param \sets \Setup(1^\seck)$
\item $(\pk,\sk) \sets \KeyGen(\param)$
%\item $(c^*,K_0) \sets \Encap_b(\pk)$
%\item $K_1 \rand \mathcal{K}$
%\item $K^* = K_b$
\item $(c^*,K^*) \sets \Encap_b(\pk)$
\item $b' \sets \A(\GUESS:c^*,K^*)$
\item \comreturn $b'$
\end{enumerate}

\end{minipage}
}
\end{tabular}

\medskip

\begin{definition}[$\INDCPA$ Security]
	A key encapsulation scheme KEM is $\INDCPA$-secure if for
        every PPT (probabilistic polynomial time) adversary $\advA$,
        we have that
        \[\AdvINDCPA{\KEM}{\advA}:=|\Pr[\INDCPAreal^\advA\Rightarrow 1]-\Pr[\INDCPArand^\advA \Rightarrow 1]|\] is negligible.
\end{definition}

\subsubsection*{Description of the scheme}
\label{subsubsec:descriptionKEM}
Our scheme contains a hash function $G$ modeled as a Random Oracle Model (ROM).
\begin{itemize}
	\item $\KeyGen(1^\seck)$: 
	\begin{itemize}
		\item choose an irreducible polynomial $P \in \Fq[X]$ of degree $n$.
		\item choose uniformly at random a subspace $F$ of $\Fqm$ of dimension $d$ and sample a couple of 	vectors $(\xv,\yv) \rand F^n \times F^n$ such that $\Supp(\xv)=\Supp(\yv)=F$.
		\item compute $\hv = \xv^{-1}\yv \mod P$.
		\item define $\pk = (\hv, P)$ and $\sk = (\xv,\yv)$.
	\end{itemize}
	\item $\Encap(\pk)$:
	\begin{itemize}
		\item choose uniformly at random a subspace $E$ of $\Fqm$ of dimension $r$ and sample a couple of 	vectors $(\ev_1,\ev_2) \rand E^n \times E^n$ such that $\Supp(\ev_1)=\Supp(\ev_2)=E$.
		\item compute $\cv = \ev_1 + \ev_2\hv \mod P$.
		\item define $K = G(E)$ and return $\cv$.
	\end{itemize}
	\item $\Decap(\sk)$:
	\begin{itemize}
		\item compute $\xv\cv = \xv\ev_1+ \yv\ev_2 \mod P$ and recover $E$ with the RSR algorithm \ref{algo:RSR}.
		\item recover $K = G(E)$.
	\end{itemize}
\end{itemize}

We need to have a common representation of a subspace of dimension $r$ of $\Fqm$. The natural way is to choose the unique matrix $\Mv \in \Fq^{r\times m}$ of size $r\times m$ in its row echelon form such that the rows of $\Mv$ are a basis of $E$.

We deal with the semantic security of the KEM in Section \ref{subsec:secu}.

\begin{figure}[!ht]
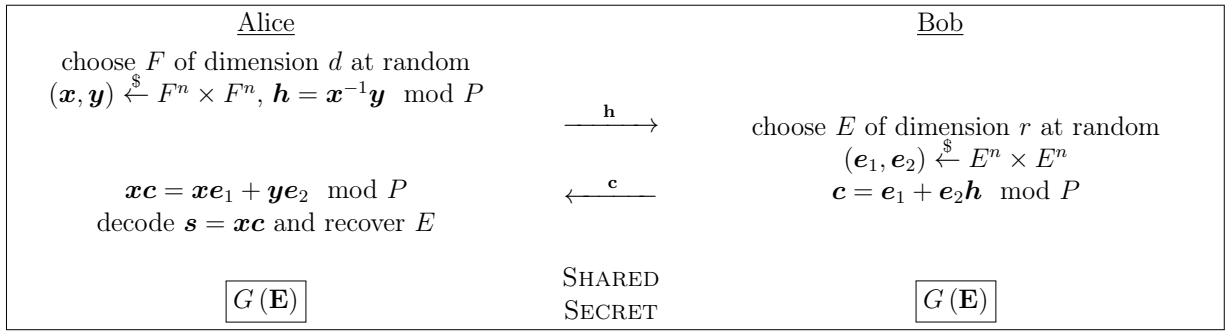


\resizebox{\textwidth}{!}{
\fbox{\hspace{-.0cm}%
\begin{minipage}{1.15\textwidth}
\begin{minipage}{.4225\textwidth}
\parbox{\textwidth}{\centering
\underline{Alice}}
\end{minipage} 
\begin{minipage}{.1125\textwidth}
\parbox{\textwidth}{}
\end{minipage} ~
\begin{minipage}{.4225\textwidth}
\parbox{\textwidth}{\centering
\underline{Bob}}
\end{minipage} \\[2mm]
\begin{minipage}{.4225\textwidth}
\parbox{\textwidth}{\centering%
choose $F$ of dimension $d$ at random\\
$(\xv,\yv) \rand F^n\times F^n$, $\hv = \xv^{-1}\yv \mod P$\\[10mm]
$\xv\cv = \xv\ev_1 + \yv\ev_2 \mod P$\\
decode $\sv = \xv\cv$ and recover $E$\\
~\\[2mm]
\fbox{$G\left(\mathbf{E}\right)$}%
}%
\end{minipage} ~
\begin{minipage}{.1125\textwidth}
\parbox{\textwidth}{\centering%
~\\[7mm]
$\xrightarrow{~~~~\mathbf{h}~~~~}$\\[6mm]
$\xleftarrow{~~~~\mathbf{c}~~~~}$\\[8mm]

%$\leftarrow$\textsc{Shared Key}$\rightarrow$
\textsc{Shared Secret}%
}
\end{minipage} ~
\begin{minipage}{.4225\textwidth}
\parbox{\textwidth}{\centering%
~\\[8mm]
choose $E$ of dimension $r$ at random\\
$(\ev_1,\ev_2) \rand E^n\times E^n$ \\
$\cv = \ev_1 + \ev_2\hv \mod P$\\[12mm]
\fbox{$G\left(\mathbf{E}\right)$}%
}
\end{minipage}
\end{minipage}
}%
}
\caption{\label{fig:new_proto}Informal description of our new Key Encapsulation Mechanism. $\mathbf{h}$ constitutes the public key.}
\end{figure}

\subsection{An IND-CCA-2 PKE based on the rank metric}
\label{subsec:PKE}
%CRS contains a hash function $G$ modeled as a ROM
%
%\begin{itemize}
%	\item KeyGen samples $\vec x, \vec y$ of small weight and outputs a public key $\vec h = \vec x^{-1} \cdot \vec y$.
%	\item To Encrypt a message $M$, one samples $\vec r_1,\vec r_2$ of support $E$ and sends $\vec c=\vec r_1 + \vec h \vec r_2$, and $\vec{d} = M \oplus G(E)$.
%	\item To Decrypt a message, one uses the secret $\vec x, \vec y$ to recovers $E$ from $\vec c$, and $M = \vec d \oplus G(E)$.
%\end{itemize}
Our KEM can be slightly modified to become a PKE cryptosystem.

A Public Key Encryption (PKE) scheme is defined by three algorithms: the key generation algorithm $\keygen$ which takes on input the security parameter $\lambda$ and outputs a pair of public and private keys $(pk,sk)$; the encryption algorithm $\encrypt(pk,M)$ which outputs the ciphertext $C$ corresponding to the message $M$ and the decryption algorithm $\decrypt(sk,C)$ which outputs the plaintext $M$.

Our PKE scheme contains a hash function $G$ modeled as a ROM.
\begin{itemize}
\item $\keygen(1^\lambda)$:
	\begin{itemize}
	\item choose an irreducible polynomial $P \in \Fq[X]$ of degree $n$.
	\item choose uniformly at random a subspace $F$ of $\Fqm$ of dimension $d$ and sample a couple of 	vectors $(\xv,\yv) \rand F^n \times F^n$ such that $\Supp(\xv)=\Supp(\yv)=F$.
	\item compute $\hv = \xv^{-1}\yv \mod P$.
	\item define $pk = (\hv, P)$ and $sk = (\xv,\yv)$.
	\end{itemize}
\item $\encrypt(pk,M)$: 
	\begin{itemize}
	\item choose uniformly at random a subspace $E$ of $\Fqm$ of dimension $r$ and sample a couple of 	vectors $(\ev_1,\ev_2) \rand E^n \times E^n$ such that $\Supp(\ev_1)=\Supp(\ev_2)=E$.
	\item compute $\cv = \ev_1 + \ev_2\hv \mod P$.
	\item output the ciphertext $C = (\cv, M \oplus G(E))$.
	\end{itemize}
\item $\decrypt(sk,C)$:
	\begin{itemize}
	\item compute $\xv\cv = \xv\ev_1+ \yv\ev_2 \mod P$ and recover $E$ with the ideal RSR algorithm \ref{algo:RSR}.
	\item output $M = C \oplus G(E)$.
	\end{itemize}
\end{itemize}

\subsection{Security of our schemes}
\label{subsec:secu}

\begin{theorem}
	Under the \emph{Ideal LRPC indistinguishability} \ref{prob:IndLRPC} and the \emph{Ideal-Rank Support Recovery} \ref{prob:I-RSR} Problems, the KEM presented above in section \ref{subsubsec:descriptionKEM} is indistinguishable against Chosen Plaintext Attack in the Random Oracle Model.
\end{theorem}

\begin{proof}
	We are going to proceed in a sequence of games. The simulator first starts from the real scheme. First we replace the public key matrix by a random element, and then we use the ROM to solve the Ideal-Rank Support Recovery.% It should be noted that the ROM is not used in a programmable way, it is just here 

	We start from the normal game $G_0$:
	We generate the public key $\hv$ honestly, and $E, \cv$ also
	\begin{itemize}
		\item 
	In game $G_1$, we now replace $\hv$ by a random vector, the rest is identical to the previous game. From an adversary point of view, the only difference is the distribution of $\hv$, which is either generated at random, or as a product of low weight vectors. This is exactly the \emph{Ideal LRPC indistinguishability} problem, hence \[\Adv_{\mA}^{G_0} \leq \Adv_{\mA}^{G_1} + \Adv_{\mA}^{\DUPKS}.\]

	\item In game $G_2$, we now proceed as earlier except we receive $\hv, \cv$ from a Support Recovery challenger. After sending $\cv$ to the adversary, we monitor the adversary queries to the Random Oracle, and pick a random one that we forward as our simulator answer to the Ideal-Rank Support Recovery problem. Either the adversary was able to predict the random oracle output, or with probably $1/q_G$, we picked the query associated with the support $E$ (by $q_G$ we denote the number of queries to the random oracle $G$), hence \[\Adv_{\mA}^{G_1} \leq 2^{-\lambda} + 1/q_G \cdot \Adv_{\mA}^{\IRSR}\] which leads to the conclusion.

\end{itemize}
\end{proof}

\begin{theorem}
	Under the \emph{Ideal LRPC indistinguishability} \ref{prob:IndLRPC} and the \emph{Ideal-Rank Support Recovery} \ref{prob:I-RSR} Problems, the encryption scheme presented in section \ref{subsec:PKE} is indistinguishable against Chosen Plaintext Attack in the Random Oracle Model.
\end{theorem}

\begin{proof}
	We are going to proceed in a sequence of games. The simulator first starts from the real scheme. First we replace the public key matrix by a random element, and then we use the ROM to solve the Ideal-Rank Support Recovery.% It should be noted that the ROM is not used in a programmable way, it is just here 

	We start from the normal game $G_0$:
	We generate the public key $\hv$ honestly, and $E, \cv$ also
		\begin{itemize}
			\item 
	In game $G_1$, we now replace $\hv$ by a random vector, the rest is identical to the previous game. From an adversary point of view, the only difference is the distribution of $\hv$, which is either generated at random, or as a product of low weight vectors. This is exactly the \emph{Ideal LRPC indistinguishability} problem, hence $$\Adv_{\mA}^{G_0} \leq \Adv_{\mA}^{G_1} + \Adv_{\mA}^{\DUPKS}.$$

	\item In game $G_2$, we now proceed as earlier except we replace $G(E)$ by random. It can be shown, that by monitoring the call to the ROM, the difference between this game and the previous one can be reduced to the Ideal-Rank Support Recovery problem, so that: \[\Adv_{\mA}^{G_1} \leq 2^{-\lambda} + 1/q_G \cdot \Adv_{\mA}^{\IRSR}.\]

	\item In a final game $G_3$ we replace $C = M \oplus \mathsf{Rand}$ by just $C = \mathsf{Rand}$, which leads to the conclusion.
\end{itemize}
\end{proof}

\subsubsection*{CCA-2 security proof}

When applying the HHK \cite{HHK17} framework for the Fujisaki-Okamoto transformation, one can show that the final transformation is CCA-2 secure such that:
$$\Adv_{\mA}^{\mathrm{CCA-2}} \leq q_G \cdot \delta + q_V \cdot 2^{- \gamma} + \frac{2 q_G + 1}{|\mathcal{M}|} + 3  \Adv_{\mA}^{\mathrm{CPA}}$$
where $q_G$ is the number of queries to the random oracle $G$ and $q_V$ is the number of verification queries.

As our scheme is CPA secure, the last term is negligible, we can handle exponentially large message space for a polynomial number of queries, so the previous is too.

As shown before, our scheme is gamma-spread so again for a polynomial number of verification queries, the term in $q_V$ is negligible.

The tricky term remaining is $q_G \cdot \delta$, this is the product
of the number of queries  to the random oracle, by the probability of
generating an decipherable ciphertext in an honest execution. For real
applications, we want schemes to be correct enough so that the
probability of such an occurrence is very small. This often leads, in
application in running with a probability of a magnitude of
$2^{-64}$. This may seem not low enough for pure cryptographic
security, however it should be noted that this number corresponds to the number of requests adversarially generated where the simulator gives an honest answer to a decryption query, which would mean that a single user would be able to do as many queries as expected by the whole targeted users in a live application, so a little trade-off at this level seems more than fair.

\subsection{Best known attacks on $\RSD$}

The complexity of practical attacks grows very fast with the size of parameters,
there is a structural reason for this: for Hamming distance a key notion
for the attacks consists of counting the number of words of length $n$ and support size $t$,
which corresponds to the Newton binomial coefficient $\binom{n}{t}$, 
whose value is exponential and upper bounded
by $2^n$. In the rank metric case, counting the number of possible supports of size
$r$ for a rank code of length $n$ over $\Fqm$ corresponds to counting the number
of subspaces of dimension $r$ in $\Fqm$: {\bf the Gaussian binomial coefficient}
of size roughly $q^{rm}$, whose value is also exponential but with a quadratic term
in the exponent.

There exist two types of generic attacks on the problem:

- {\bf combinatorial attacks}: these attacks are usually the best ones for small values
of $q$ (typically $q=2$) and when $n$ and $k$ are not too small,
when $q$ increases, the combinatorial aspect makes them less efficient.

The first non-trivial attack on the problem was proposed by Chabaud
and Stern \cite{CS96} in 1996.
It was improved in 2002 by Ourivski and Johannson \cite{OJ02} who proposed
a new attack in $\OO{(n-k)^3m^3q^{(w-1)(k+1)}}$. However, these two attacks did not take account of the value of $m$ in the
exponent. Very recently the two previous attacks were generalized in \cite{GRS16} (and used
to break some repairs of the GPT cryposystems) moreover an algebraic new setting
was also proposed, which gives an attack in $\OO{w^3k^3q^{w\Ceiling{\frac{(w+1)(k+1)-n-1}{w}}}}$. Finally, a last improvement of the combinatorial attack of \cite{GRS16} was proposed this year in \cite{AGHT17}, for a complexity of $\OO{(n-k)^3m^3q^{w\Ceiling{\frac{(k+1)m}{n}}-m}}$.

- {\bf algebraic attacks}: the particular nature of the rank metric makes it a natural field
for algebraic attacks and solving by Groebner basis, since these attacks 
are largely independent of the value of $q$
and in some cases may also be largely independent of $m$.
These attacks are usually the most efficient ones when $q$ increases.
There exist different types of algebraic equations settings to try to solve
a multivariate system with Groebner basis.
The algebraic context proposed by Levy and Perret \cite{LP06} in 2006 
considers a quadratic setting over $\Fq$ 
by taking as unknowns the support $E$ of the error and the error
coordinates regarding $E$. It is also possible to consider
the Kernel attack by \cite{FLP08} and the minor approach \cite{FSS10SV} 
which give multivariate equations of degree $r+1$ over $\Fq$ obtained from minors of matrices
Finally,
the recent annulator setting by Gaborit et {\it al.} in \cite{GRS16}
(which is valid on certain type of parameters but may not be independent on $m$)
give multivariate sparse equations of degree $q^{r+1}$ but on the large field
$\Fqm$ rather than on the base field $\Fq$. The latter attack is based on the
notion of $q$-polynomial \cite{O33} and is particularly efficient when $r$ is small.
Moreover, all these attacks can be declined in an hybrid approach where some
unknowns are guessed. In practice (and for the case of this paper)
for small $q$ these attacks are less efficient than combinatorial attacks.

\subsubsection{Impact of quantum algorithms on the complexity of the attacks}
In post-quantum cryptography, it is important to study the impact of quantum computer on the security of the schemes. In rank metric cryptography, the problems on which we are based are still difficult for a quantum computer, however there is a quantum speedup for combinatorial attacks: the exponent of the complexity is divided by 2 \cite{GHT16}, which lead to a complexity of $\OO{(n-k)^3m^3q^{\frac{w}{2}\Ceiling{\frac{(k+1)m}{n}}-m}}$.

Concerning the algebraic attacks, there is currently no quantum speedup as far as we know.

\subsection{Parameters}
\label{subsec:Parameters}
In this section, we give some sets of parameters for a security parameters of $128$, $192$ and $256$. In all cases, we have chosen $q = 2$. The choice of these parameters depends on the probability of decryption failure and the complexity of the attacks against our cryptosystem:
\begin{enumerate}
\item Message attacks: they consist in solving an instance of the $\IRSR$ problem of weight $r$ \ref{prob:I-RSR}.
\item Attacks on the secret key: they consist to recover the structure of the Ideal-LRPC code by computing a codeword of weight $d$ in the dual code \ref{prob:IndLRPC}.
\item Spurious key attack: as in the NTRU case (see \cite{HPS98})
this attack corresponds to finding small word
vectors in $\Hv$ with rank slightly greater than $d$, and to use
them for decoding. Although theoretically possible, this attack is not doable 
in practice since the fact that $\Hv$ contains small weight vectors implies
that many words of weight $2d$ exist. We do not go into details in this article
but as for MDPC codes \cite{MTSB12}, when the weight increases the complexity of the attacks
grows faster than the number of small weight vectors, so that this attacks - as for NTRU and MDPC- does not work in practice. 
\end{enumerate}
The different parameters are :
\begin{itemize}
\item $n$ is the length the vectors $\hv$ and $\cv$ sent on the public channel.
\item $m$ is the degree of the extension $\Fm$. 
\item $d$ is the weight of the ideal LRPC code used in the protocol.
\item $r$ is the weight of the error.
\item $P$ is the irreducible polynomial of degree $n$ of $\F[X]$ which defines the ideal LRPC code. We have chosen sparse polynomials in order to diminish the computation costs. 
\item the structural attack parameter is the complexity of the best attack to recover the structure of the ideal LRPC code. It consists in looking for a codeword of weight $d$ in an ideal LRPC of type $[2n,n]_{2^m}$ defined by the parity-check matrix $(\Iv_n | \Hv)$ where
\[
\Hv = \begin{pmatrix}
\hv(X) \mod P \\
X\hv(X) \mod P \\
\vdots \\
X^{n-1}\hv(X) \mod P
\end{pmatrix}.
\]

Because of the structure of the ideal LRPC, the complexity of looking for a codeword of weight $d$ is easier than in a random code (see \cite{ABDGHRTZ17a}). The complexity of the attack, for these parameters, is $\OO{(nm)^\omega 2^{d\Ceiling{\frac{m}{2}} - m - n}}$ where $\omega = \log_2 7 \approx 2.8$ is the cost of linear algebra.

\item the generic attack parameter is the complexity of the best attack to recover the support $E$. It consists to solve the I-RSD problem for a random ideal code of type $[2n,n]_{2^m}$ and an error of weight $r$. For our parameters, the complexity of this attack is $\OO{(nm)^\omega 2^{r\Ceiling{\frac{m(n+1)}{2n}} - m}}$ where $\omega = \log_2 7 \approx 2.8$ is the cost of linear algebra.
\item $p_f$ is the probability of failure of the decoding algorithm. We have chosen the parameters such that the theoretic upper bound is below $2^{-26}$ for the KEM. For the PKE, it is necessary to have a lower probability of failure since the public key can be used to encrypt a high number of messages. We give two sets of parameters: the first one for a probability below $2^{-64}$ and the second one for a probability below $2^{-80}$.
\item the entropy parameter is the entropy of the subspace $E$. It is equal to $\log_2 \left( \cg{r}{m} \right)$ and has to be greater than the security parameter. We represent $E$ by a matrix of size $r \times m$ in its row echelon form.
%\item the exchanged data is the number of bits sent by Alice and Bob during the protocol. Since each of them sends one vector of $\Fm^n$, it is equal to $2mn$.
\item the public key size is the number of bits needed to represent the public key $\hv \in \Fqm^n$. It is equal to $mn$.
\end{itemize}

{\bf Comparison with other cryptosystems:} In terms of size of public key,
the system compares very well with other proposal to the NIST competition. 
For 128 bits security key exchange
(in which case one is allowed to a not too small decryption failure rate)
the BIKE proposal (MDPC cryptosystem) has a public key of 10 kilo bits, hence 
the LRPC cryptosystem has public keys roughly three times smaller.
In terms of encryption (with very small or null decryption failure rate), 
the system is comparable to the RQC cryptosystem,
and far smaller than the McEliece cryptosystem.
Globally among all candidates to the NIST competition, the LRPC cryptosystem
has the second smaller size of public keys for 128 bits security, better 
than candidates based on lattices, but a litle larger than 
systems based on isogenies. Moreover the system is also efficient in terms
of speed compared to the other NIST proposals.

\begin{center}
\begin{table}[h]\label{table:ParameterKEM}
\begin{tabular}{|c|c|c|c|}
\hline
Security			&	128				&	192					&	256 \\
\hline
$n$ 				&	47				&	53					&	67\\
\hline
$m$ 				&	71				&	89					&	113\\
\hline
$d$					&	6				&	7					& 	8\\
\hline
$r$					&	5				&	6					&	7\\
\hline
$P$ 				& $X^{47}+X^5+1$	& $X^{53}+X^6+X^2+X+1$	& $X^{67}+X^5+X^2+X+1$\\
\hline
Structural Attack 	&	130				&	207					&	312\\
\hline
Generic Attack 		&	146				&	221					&	329\\
\hline
$p_f$ 				&	$2^{-30}$		&	$2^{-32}$			&	$2^{-36}$\\
\hline
Entropy				&	311				&	499					&	743\\
\hline
Public key (bits)	&	3337			&	4717				&	7571\\
\hline
\end{tabular}
\caption{Examples of parameters of our KEM \ref{subsec:KEM}}
\end{table}

\begin{table}[h]\label{table:ParameterProbaFailure-64}
\begin{tabular}{|c|c|c|c|}
\hline
Security			&	128				&	192					&	256 \\
\hline
$n$ 				&	83				&	83					&	89\\
\hline
$m$ 				&	71				&	101					&	107\\
\hline
$d$					&	7				&	7					& 	8\\
\hline
$r$					&	5				&	5					&	6\\
\hline
$P$ 				& $X^{83}+X^7+X^4+X^2+1$	& $X^{83}+X^7+X^4+X^2+1$	& $X^{89}+X^{38}+1$\\
\hline
Structural Attack 	&	133				&	209					&	273\\
\hline
Generic Attack 		&	144				&	195					&	260\\
\hline
$p_f$ 				&	$2^{-64}$		&	$2^{-64}$			&	$2^{-64}$\\
\hline
Entropy				&	331				&	481					&	607\\
\hline
Public key Size (bits)	&	5893			&	8383			& 9523\\
\hline
\end{tabular}
\caption{Example of parameters of our PKE \ref{subsec:PKE}, $p_f \leqslant 2^{-64}$}
\end{table}

\begin{table}[t]\label{table:ParameterProbaFailure-80}
\begin{tabular}{|c|c|c|c|}
\hline
Security			&	128				&	192					&	256 \\
\hline
$n$ 				&	101				&	103					&	103\\
\hline
$m$ 				&	79				&	97					&	107\\
\hline
$d$					&	7				&	8					& 	8\\
\hline
$r$					&	5				&	6					&	6\\
\hline
$P$ 				& $X^{101}+X^7+X^6+X+1$	& $X^{103}+X^9+1$	& $X^{103}+X^9+1$\\
\hline
Structural Attack 	&	136				&	229					&	259\\
\hline
Generic Attack 		&	157				&	234					&	260\\
\hline
$p_f$ 				&	$2^{-80}$		&	$2^{-80}$			&	$2^{-80}$\\
\hline
Entropy				&	371				&	547					&	607\\
\hline
Public key Size (bits)	&	7979		&	9991				& 	11021\\
\hline
\end{tabular}
\caption{Example of parameters of our PKE \ref{subsec:PKE}, $p_f\leqslant 2^{-80}$}
\end{table}
\end{center}

%\textbf{Quantum speedup:} 

We also provide concrete timings of our implementations. The benchmarks were performed on an Intel\textregistered Core\texttrademark i7-4700HQ CPU running @ up to 3.40GHz and the software was compiled using GCC (version 6.3.0) with the following command : \texttt{g++ -O2 -pedantic -Wall -Wextra -Wno-vla}.

\begin{table}[h] \label{table:perfmsKEM}
\begin{center}
\begin{tabular}{|c|c|c|c|}
\hline
Security & \KeyGen & \Encap & \Decap\\
\hline
128   &  0.65  & 0.13  & 0.53\\
\hline
192  &  0.73  & 0.13  & 0.88\\
\hline
256 &  0.77  & 0.15  & 1.24\\
\hline
\end{tabular}
\caption{Timings (in ms) of our KEM \ref{subsec:KEM}}
\end{center}
\end{table}

\begin{table}[h] \label{perfcyclesKEM}
\begin{center}
\begin{tabular}{|c|c|c|c|}
\hline
Security & \KeyGen & \Encap & \Decap\\
\hline
128   &  1.58  & 0.30  & 1.27\\
\hline
192  &  1.74  & 0.31  & 2.09\\
\hline
256 &  1.79  & 0.35  & 2.89\\
\hline
\end{tabular}
\caption{Timings (in millions of cycles) of our KEM \ref{subsec:KEM}}
\end{center}
\end{table}

\begin{table}[h] \label{perfmsPKE}
\begin{center}
\begin{tabular}{|c|c|c|c|c|}
\hline
Security & $p_f$ & \KeyGen & \Encap & \Decap\\
\hline
128 & 64   &  1.09  & 0.22  & 1.04\\
\hline
192 & 64  &  1.33  & 0.23  & 1.08\\
\hline
256 & 64 &  1.51  & 0.25  & 1.58\\
\hline
128 & 80 &  1.51  & 0.29  & 1.16\\
\hline
192 & 80 &  1.82  & 0.36  & 1.80\\
\hline
256 & 80 &  1.94  & 0.31  & 1.68\\
\hline
\end{tabular}
\caption{Timings (in ms) of our PKE \ref{subsec:PKE}}
\end{center}
\end{table}

\begin{table}[h] \label{perfcyclesPKE}
\begin{center}
\begin{tabular}{|c|c|c|c|c|}
\hline
Security & $p_f$ & \KeyGen & \Encap & \Decap\\
\hline
128 & 64   &  2.71  & 0.55  & 2.57\\
\hline
192 & 64  &  3.19  & 0.54  & 1.08\\
\hline
256 & 64 &  3.58  & 0.60  & 3.77\\
\hline
128 & 80 &  3.72  & 0.71  & 2.86\\
\hline
192 & 80 &  4.36  & 0.86  & 4.32\\
\hline
256 & 80 &  4.68  & 0.75  & 4.06\\
\hline
\end{tabular}
\caption{Timings (in millions of cycles) of our PKE \ref{subsec:PKE}}
\end{center}
\end{table}
\newpage

%% file: proof_impact_m.tex
\section{Proof of proposition \ref{prop:impact_m}} \label{app:impact_m}

\begin{proof}
We have:
\begin{IEEEeqnarray}{RCCCl}
\dim (F_1 \cap F_2) & \geqslant & \dim (E \cap F_1 \cap F_2) & = & \dim (E \cap F_1) + \dim (E \cap F_2) \\&&&&\hspace{2.4cm}- \dim \big( (E \cap F_1) + (E \cap F_2) \big) \nonumber \\
&&&=&2r-c_1-c_2-\dim (E \cap F_1 + E \cap F_2) \nonumber \\
&&&\geqslant &  2r-c_1-c_2-\dim E = r-c_1-c_2. \nonumber
\end{IEEEeqnarray}

For the second point, we need to define the projection over the quotient vector space $\Fqm/E$.
\begin{IEEEeqnarray}{LCCCC}
\text{Let }\varphi & : & \Fqm & \rightarrow & \Fqm/E \nonumber \\
&& \xv & \mapsto & \xv + E \nonumber
\end{IEEEeqnarray}
$\varphi$ is a linear map of $\Fq$ vector spaces. For all $U
\subset \Fqm, U \subset E \iff \varphi(U) = \{0\}$. Moreover we
have the equality $\dim \varphi(U) = \dim U -\dim (E\cap U)$. Therefore,
\[ \left\lbrace \begin{array}{l}
\dim \varphi(F_i) = d_i -r +c_i \nonumber \\
F_1 \cap F_2 \subset E \iff \varphi(F_1 \cap F_2) = \{0\} \nonumber
\end{array}\right. .
\]
Since $\varphi(F_1 \cap F_2) \subset \varphi(F_1)\cap \varphi(F_2)$, $\varphi(F_1)\cap \varphi(F_2) = \{0\} \implies \varphi(F_1 \cap F_2) = \{0\}$. So
\[
\Prob(F_1\cap F_2 \subset E) \geqslant \Prob(\dim (\varphi(F_1)\cap \varphi(F_2)) = 0).
\]
The probability that two subspaces of dimension $a$ and $b$ of a
vector space of dimension $n$ have a null intersection is given by the
formula
$$\frac{\cg{a}{n-b}q^{ab}}{\cg{a}{n}}$$
(see \cite{H14} Proposition 3.2). By taking $n=m-r, a = d_2-r+c_2$ and $b=d_1-r+c_1$, we obtain
\[
\prob{(F_1\cap F_2) \subset E} \geqslant \frac{\cg{d_2-r+c_2}{m-d_1-c_1}}{\cg{d_2-r+c_2}{m-r}}q^{(d_1-r+c_1)(d_2-r+c_2)}.
\]
To compute an approximation, let us use the formula $\cg{a}{n} = \prod_{i=0}^{a-1} \frac{q^n-q^i}{q^a-q^i}$. \label{DvptLimiteProba}
\begin{IEEEeqnarray}{rCCCl}
P_{a,b}(n) &=& \frac{\cg{a}{n-b}q^{ab}}{\cg{a}{n}} & = &  \frac{\displaystyle{\prod_{j=0}^{a-1} \frac{q^{n-b}-q^j}{q^a-q^j}}}{\displaystyle{\prod_{i=0}^{a-1} \frac{q^n-q^i}{q^a-q^i}}}\ q^{ab}\nonumber \\ 
& = & \left(\prod_{i=0}^{a-1} \frac{q^{n-b}-q^i}{q^n-q^i}\right)q^{ab} & = & \prod_{i=0}^{a-1} \frac{q^n-q^{i+b}}{q^n-q^i}  \nonumber \\
 & = & \prod_{i=0}^{a-1} \frac{1-q^{i+b-n}}{1-q^{i-n}}. \nonumber
\end{IEEEeqnarray}
Therefore,
\[
 \ln P_{a,b}(n) = \sum_{i=0}^{a-1} \ln (1-q^{i+b-n}) - \ln (1-q^{i-n}).
\]
$\ln(1-x) = -x+\OO{x^2}$ when $x\rightarrow 0$ hence
\begin{IEEEeqnarray}{rCl}
\ln P_{a,b}(n) & = & \sum_{i=0}^{a-1} \left(q^{i-n} - q^{i+b-n}\right) + \OO{q^{-2n}} \nonumber \\
 &= & q^{-n}(1-q^b) \sum_{i=0}^{a-1} q^i +  \OO{q^{-2n}} \nonumber \\
 & =& -q^{-n}(1-q^b)\frac{q^a-1}{q-1} + \OO{q^{-2n}} \nonumber
\end{IEEEeqnarray}
when $n\rightarrow +\infty$.\\
$e^x = 1+x+\OO{x^2}$ when $x \rightarrow 0 \implies P_{a,b}(n) = 1-\frac{q^{-n}(q^a-1)(q^b-1)}{q-1}+\OO{q^{-2n}}$
By replacing the variables, we obtain
\[
\Prob(F_1\cap F_2 \subset E) \geqslant 1-\frac{q^{-m+r}(q^{d_1-r+c_1}-1)(q^{d_2-r+c_2}-1)}{q-1} + \OO{q^{-2m}}.
\]

\end{proof}